\documentclass[12pt]{article}
\usepackage{amsmath,amsthm,amssymb}
\usepackage[margin=1in]{geometry}
\usepackage[numbers,square,sort]{natbib}
\usepackage{mathrsfs}
\usepackage{booktabs}
\usepackage{nicefrac}
\usepackage{array}
\usepackage{varwidth}
\usepackage{thm-restate}
\usepackage{bbm}
\usepackage{graphicx}
\usepackage{subcaption}
\usepackage{caption}
\usepackage{enumitem}
\usepackage{wrapfig}
\usepackage{cancel}

\usepackage{tcolorbox}
\definecolor{systemcolor}{HTML}{A9A9A9}
\definecolor{usercolor}{HTML}{2980B9}
\definecolor{assistantcolor}{HTML}{27AE60}
\newtcolorbox{systemmessagebox}{
  colback=systemcolor!10, colframe=systemcolor,
  sharp corners, boxrule=1pt, left=5pt, right=5pt, top=5pt, bottom=5pt,
  fonttitle=\bfseries, title=System
}

\newtcolorbox{usermessagebox}{
  colback=usercolor!10, colframe=usercolor,
  sharp corners, boxrule=1pt, left=5pt, right=5pt, top=5pt, bottom=5pt,
  fonttitle=\bfseries, title=User
}

\newtcolorbox{assistantmessagebox}{
  colback=assistantcolor!10, colframe=assistantcolor,
  sharp corners, boxrule=1pt, left=5pt, right=5pt, top=5pt, bottom=5pt,
  fonttitle=\bfseries, title=Assistant
}

\usepackage{hyperref}
\usepackage{cleveref}
\usepackage{xcolor}

\newtheorem{theorem}{Theorem}[section]
\newtheorem{lemma}[theorem]{Lemma}

\newtheorem{remark}[theorem]{Remark}
\newtheorem{proposition}[theorem]{Proposition}

\theoremstyle{definition}

\Crefname{assumption}{Assumption}{Assumptions}

\title{The Backfiring Effect of Weak AI Safety Regulation}
\author{Benjamin Laufer\thanks{Cornell Tech, \texttt{bdl56@cornell.edu}.},\  \ Jon Kleinberg\thanks{Cornell University.%, \texttt{kleinberg@cornell.edu}
}\textsuperscript{\ ,\S},\  \ Hoda Heidari\thanks{Carnegie Mellon University.%, \texttt{hheidari@andrew.cmu.edu}
}\textsuperscript{\ ,}\thanks{Equal contribution.}}
\date{}
\newenvironment{acks}
{
  \section*{Acknowledgements.}
}
{
The authors thank the AI, Policy and Practice (AIPP) group at Cornell University and the Digital Life Initiative (DLI) at Cornell Tech. We thank James Grimmelmann, Aleksandra Korolova, Karen Levy, Hamidah Orderinwale, Helen Nissenbaum, Tori Qiu, and Manish Raghavan for helpful conversations and comments.

This work is supported by a grant from the John D. and Catherine T. MacArthur Foundation. Ben Laufer is also supported by a doctoral fellowship from DLI, a LinkedIn-Bowers CIS PhD Fellowship,  and a SaTC NSF grant CNA-1704527. Jon Kleinberg is also supported by a Vannevar Bush Faculty Fellowship, AFOSR aware FA9550-19-1-0183, and a grant from the Simons Foundation. Hoda Heidari acknowledges support from NSF (IIS-2040929 and IIS-2229881) and PwC (through the Digital Transformation and Innovation Center at CMU). Any opinions, findings, conclusions, or recommendations expressed in this material do not reflect the views of NSF or other funding agencies.
}

\begin{document}
\maketitle
\begin{abstract}
Recent policy proposals aim to improve the safety of general-purpose AI, but there is little understanding of the efficacy of different regulatory approaches to AI safety. We present a strategic model that explores the interactions between safety regulation, the general-purpose AI technology creators, and domain specialists--those who adapt the technology for specific applications. Our analysis examines how different regulatory measures, targeting different parts of the AI development chain, affect the outcome of this game. In particular, we assume AI technology is characterized by two key attributes: \emph{safety} and \emph{performance}. The regulator first sets a minimum safety standard that applies to one or both players, with strict penalties for non-compliance. The general-purpose creator then invests in the technology, establishing its initial safety and performance levels. Next, domain specialists refine the AI for their specific use cases, updating the safety and performance levels and taking the product to market. The resulting revenue is then distributed between the specialist and generalist through a revenue-sharing parameter. Our analysis reveals two key insights: First, weak safety regulation imposed predominantly on domain specialists can backfire. While it might seem logical to regulate AI use cases, our analysis shows that weak regulations targeting domain specialists alone can unintentionally reduce safety. This effect persists across a wide range of settings. Second, in sharp contrast to the previous finding, we observe that stronger, well-placed regulation can in fact mutually benefit \emph{all} players subjected to it. When regulators impose appropriate safety standards on both general-purpose AI creators and domain specialists, the regulation functions as a commitment device, leading to safety and performance gains, surpassing what is achieved under no regulation or regulating one player alone. 
\end{abstract}

\section{Introduction}

 As Generative Artificial Intelligence (AI) and related technologies gain traction, there is an increasing number of proposals for regulation to improve safety. Many of these proposals must at some level grapple with the following question:
 Who should be targeted with AI regulation--the producers of general-purpose AI models\footnote{Such AI models are at times referred to as ``foundation'' or ``frontier'' models~\citep{bommasani2021opportunities,anderljung2023frontier}. Throughout this paper, we will use the technology of general-purpose AI to refer to large-scale models that can be adapted to a wide range of tasks and domains.} or the domain-specialists who adapt the technology for specific use cases? There are seemingly reasonable positions that favor regulating one entity, the other, both, or neither. For example, the downstream domain specialists and deployers are some of the last entities to exert influence on the technology before it interacts with consumers directly, so it is perhaps reasonable that regulation for consumer safety might target requirements at these entities. In contrast, the upstream entities developing general-purpose models exert impact on these models earlier in their development trajectories, facilitating or hindering downstream adoption, which might justify certain regulatory requirements including disclosure mandates \cite{longpre2025house} and liability standards. Of course, even regulations that solely target one of these actors might impact the other, because their incentives and decisions are intertwined.

We have seen variants of these debates play out as different jurisdictions and policymakers have proposed various regulatory approaches to AI. A number of existing regulation proposals leverage the observation that AI is developed by multiple, interacting actors. Examples include Colorado's AI Act, California's Senate Bill 1047, and the EU AI Act. These frameworks attempt to define the relevant actors, such as base developers and downstream deployers, in order to design conditions and stipulations for determining whether and to whom liability standards, disclosure requirements, or other interventions apply. These conditions and stipulations vary across proposals and policies, with possibly significant implications for the incentives of the players involved in the development of AI technologies and applications. 

\textbf{Modeling the impact of regulatory regimes on AI performance and safety.} Given that there are a range of different possible approaches to targeting AI regulation and assessing the impact of each alternative empirically is prohibitive, formal models can enable reasoning about the various regulatory impacts. This paper puts forward a strategic model of the interactions between a general-purpose technology producer ($G$) and a domain specialist ($D$), building on the ``fine-tuning games'' model proposed by \citet{laufer2024fine}. As the two actors develop an AI technology, they each decide whether and how to invest in two key attributes of technology: \textit{performance}, denoted by $\alpha$, and \textit{safety}, denoted by $\beta$. We assume these actors are operating in a market; each actor experiences some cost for their investment in safety and performance, and obtains a share of the revenue out of the deployment of the AI product/service in the market. 

To provide some intuition for what this investment pattern might look like, imagine a firm, $G$, producing a general-purpose language model that may be used in three domains -- say, by healthcare providers ($D_1$), law firms ($D_2$), and financial services ($D_3$). The general-purpose developer moves first, and in light of the particular costs she faces and the anticipated responses from the downstream players, she chooses a certain strategy, represented by a pairing of performance and safety investments $(\alpha_0,\beta_0)$. Once this investment has been made, the attributes of the technology at this stage can be thought of as akin to a `base camp,' from which domain specialists may choose to climb further by investing their own effort toward improving the technology's safety and/or performance in their respective domains. Of course, each domain faces their own delicate balance of safety risks and performance costs, so the ultimate safety and performance pairs $(\alpha_i,\beta_i)$ ($i=1,2,3$) differ across the three domains. See Figure \ref{fig:teaser} (a) for a visualization of the investment decisions make by $G$, $D_1$, $D_2$, and $D_3$.

Equipped with this intuition about how these actors behave in an unregulated market, we now turn to our notion of regulation. We conceive of regulation as shaping the game in which players choose their strategies. In particular, this paper will focus specifically on safety regulation. We assume regulation imposes a constraint in the form of a lower bound on the players' choice of safety investment (i.e., $\beta_i$'s). If a player does not meet the regulatory lower bound on safety, they will be penalized. 
This regulatory regime can be described using two parameters $(\theta_G, \theta_D)$, representing the set of thresholds constraining the strategy space of $G$ and $D$, respectively. The regulation can target the domain-specialist only ($\theta_G=0, \theta_D>0$), the generalist only ($\theta_G=\theta_D>0$), both players ($\theta_D>\theta_G>0$), or neither ($\theta_G=\theta_D=0$). In addition to the decision of who to target, of course, the regulation encodes a decision about what level to set the safety standards. Smaller values of $\theta$ are less costly to comply with, and hence capture weaker safety requirements. 

\begin{figure}
    \centering
    \includegraphics[width=0.6\linewidth]{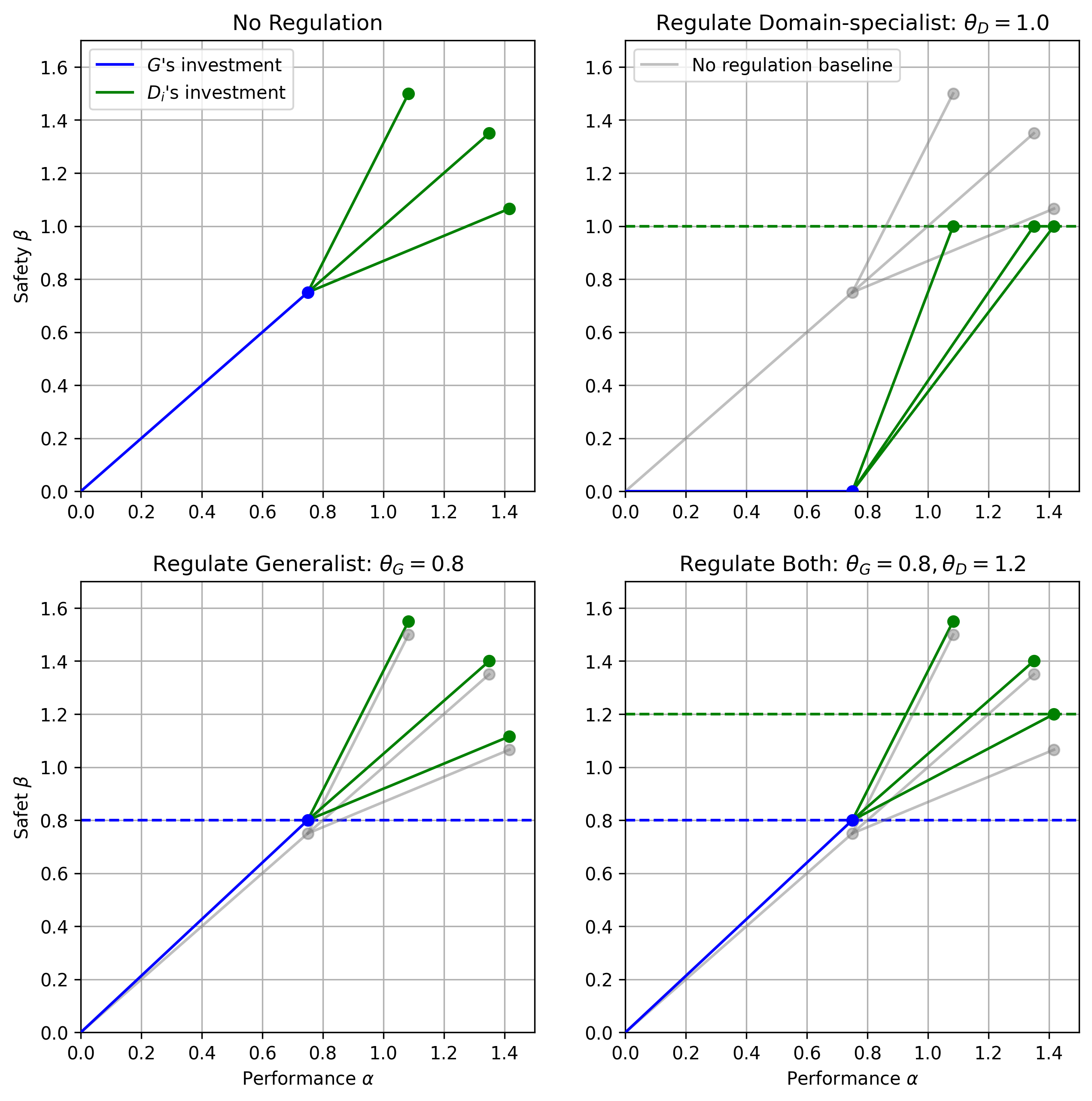}
    \caption{A depiction of a particular instance of our game-theoretic model. This instance of the game consists of one general-purpose producer and three domain-specialists. Each player has a different utility in performance-safety space which dictates the path of development. The no-regulation game (upper left) reveals the players' investment efforts when no floor is imposed on safety. Regulating the domain-specialist alone (upper right) exhibits \textit{backfiring} for all three domains, meaning the regulated safety level is lower than it would be without regulation. In this particular example, the same floor is assumed for all three domain-specialists. 
    Regulating the generalist alone (lower left) improves the safety level slightly across all three domains, compared to no-regulation. Finally, a regime that targets both generalist and specialists with regulation (lower right) is able to 1) retain the improved safety performance from regulating the generalist, 2) improve the safety level of least-safe domain-specialist, while 3) avoiding backfiring.
    }
    \label{fig:teaser}
\end{figure}

\textbf{First insight: Weak safety regulation can backfire.} Turning back to our example in Figure \ref{fig:teaser}, we observe that something striking happens in the second panel, which depicts a scenario where regulation is targeted at the domain-specialist. In this scenario, the safety investment has gotten worse. How could safety regulation -- a simple floor dictating a minimum investment level -- lead to a less safe product? The mechanism leading to this phenomenon arises because the generalist $G$ is aware of the regulatory safety requirements imposed on domain-specialists, and can use it to her advantage. When the regulator requires that a technology meets a certain level of safety investment by the time it reaches the market, the generalist has an opportunity to engage in a sort of \textit{free-riding} behavior. The generalist is comfortable setting up the base camp at lower altitude, because she knows that the domain-specialist nonetheless has to climb to a level of investment that complies with regulation.\footnote{Of course, there may be scenarios where the regulation is too high, or the generalist's investment is too low, such that the domain specialist gives up entirely and abstains from any involvement in the technology's production. We handle these scenarios in our analysis of the model, but our point here is to demonstrate that there are circumstances in which regulation gives the generalist leverage to free-ride by lowering the initial investment in safety.} 

The scenario described above depicts a single instance of a more general phenomenon, which we describe as regulatory \textit{backfiring}. A safety regulation backfires if it yields a total investment in safety lower than the safety investment achieved with no regulation. We identify a number of properties of this phenomenon -- for example, backfiring only occurs when the regulation is \textit{weak}, meaning the floor on safety is at or below the level reached in the absence of regulation. Backfiring can occur when $D$ is targeted with regulation or when both $G$ and $D$ are targeted with regulation, but does not occur when only $G$ is targeted. Our results suggest that this non-monotonic effect of regulation occurs for a broad set of games with different cost and revenue functions. 
Analytically, we prove that backfiring occurs for all 
quadratic-cost games in which the players invest any non-zero amount in both performance and safety without regulation (these sets of games are defined more formally in Section \ref{subsec:multi-attribute-with-regulation-soln}). 

\textbf{Second insight: Properly-placed safety regulation can improve the technology and the players' utilities.} While weak regulation targeted predominantly at the domain-specialist can backfire, our results suggest that other regulatory regimes fare better. When safety standards are directed at both $G$ and $D$ with appropriate strength, regulation can improve not just safety, but the utilities of both players, defined as their revenue share minus their investment cost. This result might seem unintuitive: Regulation only reduces the set of choices available to each actor in our model, so how can regulation lead to choices that mutually benefit both generalist and specialist? What is stopping the players from choosing utility-optimal strategies in the absence of regulation? The reason this phenomenon occurs is a Prisoner’s Dilemma-style result: The players’ unregulated strategies, which are chosen to maximize their individual utility, fail to yield the strategies that that are globally optimal for both players. By constraining the actors away from the strategies that enable this kind of selfish behavior, regulation can act as a commitment device. The generalist can increase her investments in safety with the assurance that the domain specialist will contribute, too, rather than free-ride off of $G$’s efforts.

\begin{figure}
    \centering
    \includegraphics[width=\linewidth]{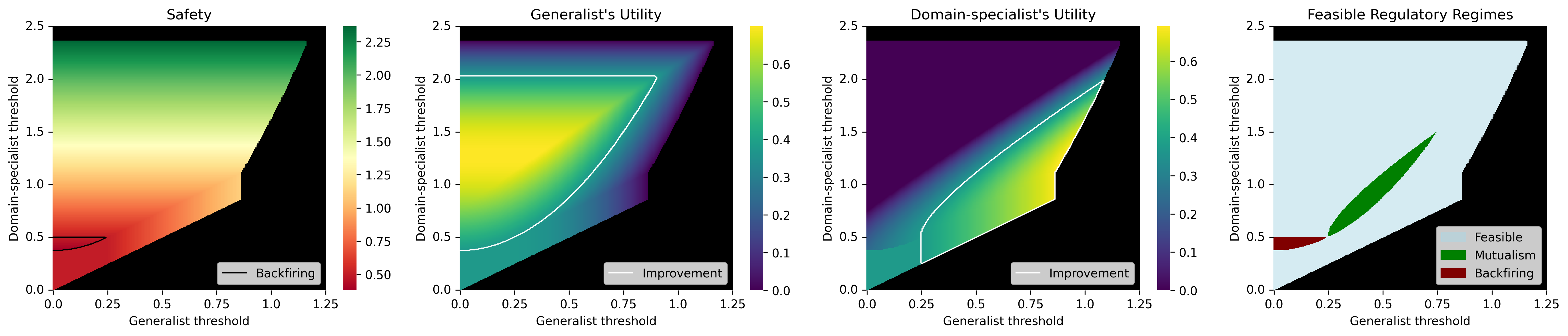}
    \caption{Simulated results for an example of two-player AI regulation model with quadratic costs. Players make costly investments in performance and safety (visualized on left), and then receive some share of revenue that comes from the total investment levels. The players' utilities -- defined as their share of the revenue minus the cost of their investment -- is visualized for the Generalist (second from left) and the Domain-Specialist (third from left). Colors represent different utility outcomes depending on different combinations of regulatory constraints $(\theta_G,\theta_D)$ which constrain the players' safety investments. The game is solved over a grid of plausible regulations: $\theta_G \in [0,1.25]$, $\theta_D \in [\theta_G,2.5]$ using increments of $0.005$, with a total of 105,651 %150,801 
    simulated regulation games. Regulations that lead the players to \texttt{abstain} are depicted in black. There exists a region where non-zero regulation yields lower safety than no regulation (highlighted on leftmost plot). There also exists a region where regulation yields improvements to each players' utility (highlighted on two center plots). The rightmost plot summarizes our results by showing the backfiring and mutualism outcomes in the $\theta_G, \theta_D$ space. Parameter values for producing the plot: $C_0 = C_1 = I_2$, $r_\alpha = r_\beta = 1$, and $\delta=0.5$.}
    \label{fig:fourpanel-separable}
\end{figure}

Games can exhibit both backfiring and mutualistic regulations, depending on who is targeted and at what threshold. For example, Figure \ref{fig:fourpanel-separable} depicts a particular instance of our game setting with one generalist and one domain-specialist. For the particular cost and revenue functions depicted, backfiring regulations and Pareto-improving regulations are possible, and the regulations yielding these effects are visualized. This figure represents a systematic sweep of all pairs of thresholds directed at the generalist, the domain specialist, or both. The pair of thresholds $(0,0)$ corresponds to the case of no regulation. The safety implications of various regulations are depicted using a red-yellow-green color scale in the leftmost plot, while the utility implications for the generalist and specialist are depicted using a purple-green-yellow color scale in the center plots. Compared to the safety and utility values at the origin points, the backfiring and Pareto-improving regions are regulations which lead to lower safety and higher utilities for both players, respectively. Although this figure depicts an example of a single game, our analysis proves that these backfiring and Pareto-improving regulations exist for a broad class of games with quadratic costs. Namely, we find that backfiring occurs in all games in which the market incentivizes some non-zero investment in both performance and safety without regulation. Our characterization of when this phenomenon occurs includes \textit{separable} scenarios (where the cost of investing in performance is independent of the cost of investing in safety), \textit{complementary} scenarios (where investing in one makes the other cheaper), and weakly \textit{interfering} scenarios (where the cost of investing in one makes the other more expensive) up to a certain bound, which we specify. We provide similar bounds for the mutualism results. 

\section{Related work}

\textbf{AI Safety Regulation.} The rise of AI-related incidents have motivated several AI incident repositories to keep track of common risks \citep{aiid,aiaaic_taxonomy}. Scholars have attempted to taxonomize AI harms to make sense of the growing array of incidents~\cite{weidinger_taxonomy_2022, shelby_sociotechnical_2023}. Some existing AI risk taxonomies organize risks primarily by \textit{domains}. These include risks to the \textit{physical} or \textit{psychological well-being} of people, \textit{human rights and civil liberties}, \textit{political and economic structures}, \textit{society and culture}, and \textit{the environment} \cite{abercrombie_collaborative_2024}. Others categorize these risks based on how they arise, including \textit{malicious use}, \textit{malfunctions}, or \textit{systemic effects} from wide adoption \cite{bengio2025international}. In our stylized model, we capture all such considerations using a single scalar that can be toggled by players through investments in safety. Common themes in policy drafts and recommendations stress the importance of balancing the goals of innovation and risk reduction, appropriately defining and targeting thresholds, and the impacts on incentives \cite{chayes2025draft,gaske2023regulation}.

\textbf{Game-theoretic models of AI development.} A line of work uses formal models to reason about the strategic and social implications of machine learning (e.g., \cite{hardt2016strategic,liu2022strategic,blum2021one,harris2021stateful,donahue2021model}). More recently, there have been proposals for using modeling approaches to understand the social and safety implications of generative AI \cite{dean2024accounting,sun2025game}. Attempts to model the development process of generative AI often make use of the observation that development is \textit{sequential} and involves \textit{multiple interacting actors} \cite{cen2023ai}. Many existing works explore different strategic aspects of the market for AI using a stackelberg game. For example, \citet{taitler2025data} use a sequential game to explore incentives for data-sharing. Further time-steps, players and decisions have been added to explore particular topics, including the level of openness and market entry dynamics \cite{xu2024economics,wu2025navigating}. \citet{taitler2025selective} introduce a particular notion of regulation in a related game-theoretic setting, and similar to our paper, they conceive of regulation as a restriction on the strategy space for developers of generative AI. Though work explicitly examining the interaction between performance and safety attributes in this setting is limited, \citet{jagadeesan2024safety} explores the interaction between these attributes in a linear regression setting in order to understand firms' market entry decisions.

\textbf{The fine-tuning games model.} 
Our work builds on and extends the \textit{fine-tuning games} model proposed in \citet{laufer2024fine}. That model builds a one-dimensional game in which players must bargain over a revenue-sharing contract before investing in performance in sequence. We extend this model in two ways: First, the players' strategy space is two-dimensional in our model, to capture the dynamic that often arises where a regulator wants to steer the technology in a direction (e.g., safety) other than that which is most-profitable (e.g., a baseline combination of performance and safety, dictated by the unregulated market). Second, we introduce the regulation, which can be seen as a \textit{floor} constraining the feasible strategy space of each player. This allows us to explore when targeting generalists, specialists, both or neither is preferable for achieving desiderata like safety.

\textbf{Economic theory and contracts.} Our work leverages pre-existing approaches that are common in the theoretical economics and game theory literatures to reason about the set of possible impacts of AI safety regulation. In particular, we draw inspiration from canonical works in contract theory \cite{grossman1992analysis, ross1973economic} and the coordination of supply chains \cite{cachon2003supply}. Our model is a variant of a Principal-Agent problem in which the strategy space is defined by two real-valued attributes, and the cost and revenue are functions of these attributes. In this way, our model draws inspiration from \citet{viscusi1993product} analyzing the possible effects of products liability schemes on innovation and safety. That model --- a one-player model with no order-of-play effects --- demonstrates that liability does not, necessarily, hamper innovation. 
We assume that innovation is sequential, meaning that an entity's investment in safety or performance builds on the contributions of past investments  \cite{bessen2009sequential, green1995division}.\footnote{However, some have observed that safety investments can \textit{degrade} as the result of fine-tuning performance investments especially when model weights are open \cite{qi2023fine, qi2024evaluating}. This scenario, and especially the interaction effects with model openness, are ripe areas for further analysis.} 
In what we call the `no-regulation' game, we assume the players revenue-share via a linear contract \cite{dutting2019simple}, a common assumption in the literature (e.g., \cite{dutting2025multi, dutting2023multi, alon2022bayesian,carroll2015robustness}). However, one way to interpret our mutualism results (Sections \ref{subsec:commitment} and \ref{subsec:mutulaism-proof}) is as a demonstration that linear contracts are sub-optimal in our setting. Our notion of regulation can be viewed as a set of non-linear contracts defined by a set of strategy constraints, and our results suggest these more expressive contracts can yield higher utility. Of course, still other forms of contracts are possible and may yield different utility implications. We leave these directions to future work.

\section{A Model of Regulating AI Safety}
\label{sec:model}

Here we offer a formal model for analyzing the effects of regulation on the development of AI applications. Our model is a sequence of sub-games between two players. Each player will choose whether and how to contribute to the technology at a certain point in the development of the technology, and some revenue is received depending on the ultimate attributes of the technology. The players are constrained by regulatory floors on safety, which will be set exogenously by a regulator. 

\textbf{Players.} A general-purpose producer, referred to as $G$, invests in a technology that may be adapted by domain-specialist(s), referred to as $D_i$. The generalist is the first to invest in the technology, meaning that before $G$ moves, the technology's attributes begin at value $0$. Each specialist $D_i$ makes an investment after the generalist has moved.

\textbf{Technology.} We say a technology is described by one or more non-negative attributes $\gamma \in \mathbb{R}^d$. In this paper, we are interested in two attributes in particular: \textit{performance} and \textit{safety}.\footnote{We note, however, that the results and inferences we draw may hold for other attributes that relate to one another with similar structure.} Unless otherwise specified, we assume $d=2$ and that $\gamma = \left[\alpha, \beta\right]$ where $\alpha$ refers to performance and $\beta$ refers to safety.

\textbf{Economic interests.} Each player, acting in a way that maximizes their self-interest, invests some non-zero amount in the technology. $G$ invests to $\gamma_0$ and each $D_i$ further invests to $\gamma_i$. Accordingly, each must pay a cost for their investment, $\phi_0(\gamma_0)$ and $\phi_{i}(\gamma_i;\gamma_0)$, respectively. After both players invest, they share a revenue that is brought in as a function of the ultimate attributes of the technology in domain $i$, $r_i(\gamma_i)$. We assume that, for some $\delta_i \in [0,1]$, $G$ gets $\delta_i r_i(\gamma_i)$ in revenue and $D_i$ gets $(1-\delta_i)r_i(\gamma_i)$. $\delta_i$ could either be exogenously fixed and given ahead of the game play, or it can be the result of bargaining between $G$ and $D_i$. When we analyze a game with only one specialist, we will drop the subscript and use $\delta$. 

\textbf{Regulation} We model regulation as imposed exogenously on the environment. Regulation is a \textit{minimum constraint} on the safety investment that the players make. A regulation that targets $G$'s investment is characterized by a value $\theta_G \in \mathbb{R}^{+}$. A non-zero regulation would constrain $G's$ strategy such that $\gamma_0[1]\geq \theta_G$. A regulation targeted at the domain-specialist, similarly, would take the form $\theta_{D}$ and lead the domain-specialist to be constrained in their strategy so $\gamma_{i}[1]\geq \theta_{D}$. 

\textbf{Gameplay.} The game proceeds as a sequence of subgames:

\begin{itemize}
    \item Regulation $\{\theta_G, \theta_{D}\}$ is announced.
    \item $G$ chooses to either \texttt{abstain} or invest in the technology, bringing it to $$\gamma_0 = \left[ \begin{array}{c}
        \alpha_0 \\
        \beta_0
    \end{array} \right].$$
    \item $D_i$ chooses to either \texttt{abstain} or invest in the technology, bringing it to $$\gamma_i = \left[ \begin{array}{c}
        \alpha_i \\
        \beta_i
    \end{array} \right].$$\\
    \item The technology brings in revenue $r_i(\gamma_i)$, which will be shared such that $G$ receives $\delta_i r_i(\gamma_i)$ and $D$ receives $(1-\delta_i)r_i(\gamma_i)$.
\end{itemize}

The utilities of the players are given below: 
$$U_G := \sum_i\delta_i r_i(\gamma_i) - \phi_0(\gamma_0); \ U_{D_i} := (1-\delta_i) r_i(\gamma_i) - \phi_i(\gamma_i;\gamma_0)$$

The best-response sub-game perfect equilibrium strategy for the generalist and specialist, respectively, can be expressed as the following optimization problems:

$$\gamma_0^* := {\arg\max}_{\gamma_0} U_G \ s.t. \beta_0\geq \theta_G; \ \ \gamma_i^* := {\arg\max}_{\gamma_i} U_{D_i} \ s.t. \beta_i\geq \theta_G.$$

Finally, the players will opt to \texttt{abstain}, if they prefer $0$ utility to any other feasible strategy. If either player chooses to \texttt{abstain}, then \textit{both} players receive $0$ utility.

\section{Closed-Form Solutions}

In this section, we analyze our regulation game where players' cost functions can be expressed as a two-degree quadratic equation. 
Specifying a quadratic function over two attributes requires defining a matrix of cost coefficients. 
The cross-terms in this matrix represent how investments the attributes interact with one another. 
For the technical portions of the paper, we use the case of one domain specialist ($D$) as our focus.
We therefore have the following cost and revenue functions:
$$\phi_0(\gamma_0) = \gamma_0^T C_0\gamma_0,$$
$$\phi_1(\gamma_1;\gamma_0) = (\gamma_1-\gamma_0)^T C_1 (\gamma_1 - \gamma_0),$$
$$r(\gamma_1) = r^T \gamma_1,
$$

$$\text{where}\ \ \ \ C_0 = \left[\begin{array}{cc}
    c_{0,\alpha\alpha} & c_{0,\alpha\beta} \\
    c_{0,\alpha\beta} & c_{0,\beta\beta}
\end{array}\right]; \ \ \ C_1 = \left[\begin{array}{cc}
    c_{1,\alpha\alpha} & c_{1,\alpha\beta} \\
    c_{1,\alpha\beta} & c_{1,\beta\beta}
\end{array}\right]; \ \ \ r= \left[\begin{array}{c}
    r_\alpha  \\
    r_\beta
\end{array}\right].$$

The players' utilities can thus be expressed as:
$$U_G:= \delta r^T \gamma_1 - \gamma_0^T C_0\gamma_0,$$
$$U_D := (1-\delta) r^T \gamma_1 - (\gamma_1-\gamma_0)^T C_1 (\gamma_1 - \gamma_0).$$
It should be noted that not all values for the above parameters correspond to realistic or interesting scenarios. For example, we assume that the diagonal entries of both cost matrices $c_{0,\alpha\alpha},c_{0,\beta\beta},c_{1,\alpha\alpha},c_{1,\beta\beta}$ are non-negative, to capture that investments in goods like safety and performance should have non-zero increasing cost. Although the cross-terms of the cost matrices can be negative, we require that $c_{0,\alpha\beta}>-\sqrt{c_{0,\alpha\alpha}c_{0,\beta\beta}}$ and $c_{1,\alpha\beta}>-\sqrt{c_{1,\alpha\alpha}c_{1,\beta\beta}}$, since it should not be that some combination of investments in $\alpha,\beta$ come at negative cost. 
Each players' choices over $\alpha$ and $\beta$ should be considered as simultaneous across the two attributes, representing a joint optimization over performance and safety. 

In this section, we start by providing sub-game perfect equilibria strategies in the case with no regulation, and then provide solutions for the regulated game. The form of problem we are dealing with is a continuous, not-necessarily-convex optimization problem with a constant number of constant-degree polynomials in a constant number of variables. Broadly, the strategy is to put forward a small number of candidate points that must be checked using a limited number of steps. These checks can be implemented numerically. 
After stating the solved subgame perfect equilibria strategies, we will move to a slate of numerical results and findings analyzing the effects of regulation. 

\subsection{Subgame perfect equilibria strategies without regulation}
\label{subsec:multi-attribute-no-regulation-soln}

In this section, we state the subgame perfect equilibrium strategies to the game under \textit{no regulation}. These can provide intuition about the behaviors in the game, before we add the additional complexity of regulation. These can be seen as a strict generalization of the Fine-Tuning Games solutions \cite{laufer2024fine} to games with two attributes that can interact. 
\begin{proposition}
    Given an AI regulation game with quadratic costs, no regulation, 
    and revenue-sharing parameter $\delta$, domain specialist $D$'s subgame perfect equilibrium strategy is one of the values in the following set:
    $$\gamma_1^* \in \left\{
    \gamma_0 + \frac{(1-\delta)}{2}C_1^{-1} r,
    \left[\begin{array}{c}
         \alpha_0 \\
         \beta_0+\frac{(1-\delta)r_\beta}{2c_{1\beta\beta}}
    \end{array}\right],
    \left[\begin{array}{c}
         \alpha_0 +\frac{(1-\delta)r_\alpha}{2c_{1\alpha\alpha}}\\
         \beta_0
    \end{array}\right],\left[\begin{array}{c}
         \alpha_0\\
         \beta_0
    \end{array}\right]
    \right\}$$
The strategy is the feasible candidate which maximizes $U_D$, subject to $U_D\geq 0, \alpha_1 \geq \alpha_0, \beta_1\geq \beta_0$.
\label{prop:D-strategy-noreg}
\end{proposition}

\begin{proposition}
    Given a two-player AI regulation game with quadratic costs, no regulation, and revenue-sharing parameter $\delta$, $G$'s best-response is one of the following candidates:
    $$\gamma_0^* \in \left\{
    \frac{\delta}{2}C_0^{-1} r,
    \left[\begin{array}{c}
         0 \\
         \frac{\delta r_\beta}{2c_{0\beta\beta}}
    \end{array}\right],
    \left[\begin{array}{c}
         \frac{\delta r_\alpha}{2c_{0\alpha\alpha}}\\
         0
    \end{array}\right],\left[\begin{array}{c}
         0\\
         0
    \end{array}\right]
    \right\}.$$
The strategy is the candidate which maximizes $U_G$, subject to $U_G\geq0, U_D\geq 0, \alpha_1 \geq 0, \beta_1\geq 0$.
\label{G-strategy-noreg}
\end{proposition}

    The proofs of the above two propositions are given in Appendix \ref{app:noregulation}. 
    The solutions offer intuition about the set of strategies players might opt to take. They may venture in the direction of some combination of performance and safety, that is, move to a point that does not reside on either constraint. Or, alternatively, they may creep along the axes constraining their strategy space, and invest minimally in either performance or safety.
    
    When do the players prefer one of these strategies over another? In general, our solutions are provided as sets of candidates because there are multiple intersecting constraints that must be checked to ensure a given candidate is optimal. 
However, our analysis reveals classes of games in which the market will lead players to invest in both safety and performance in conjunction under no regulation. We make this claim formal below. 

\begin{remark}
    Given the AI regulation game with quadratic costs, no regulation, and revenue-sharing parameter $\delta \in \left(0,1\right)$. If any player $p$'s cost interaction term satisfies the following inequalities: $$c_{p,\alpha\beta} < \min\left(\sqrt{c_{p,\alpha\alpha}c_{p,\beta\beta}},\frac{c_{p,\alpha\alpha} r_\beta}{r_\alpha}, \frac{c_{p,\beta\beta}r_\alpha}{r_\beta}\right),$$ then their best-response strategy includes non-zero investment in both performance and safety. 
\label{remark:unconstrained-condition}
\end{remark}

This claim is proven in Appendix \ref{app:non-zero-investment-no-reg}. The broad intuition is that the first of these inequalities establishes the costs are strictly convex, and the second two ensure that the player's cost interactions are not so positive that investing in both performance and safety is prohibitively expensive compared to investing in one or the other alone. The claim offers some intuition for when a player prefers to invest in both attributes together, even without regulation pushing them to invest in safety. It covers all games in which the cost interactions are negative, which we call the \textit{complementary scenario}, meaning it is cheaper to invest in both performance and safety together than to invest in each individually. It further covers all games in which the cost interactions are zero, which we call the \textit{separable scenario}, meaning there is no benefit or loss to investing in both attributes in conjunction. Finally, it covers certain instances where the cost interactions are positive, which we call the \textit{interfering scenario}, meaning safety investments make performance more costly, and vice versa.

\subsection{Subgame perfect equilibria strategies with regulation}
\label{subsec:multi-attribute-with-regulation-soln}

Here we provide the subgame perfect equilibria strategies of the two players in our two-attribute game, in the presence of regulation. Notice that the no-regulation gameplay can be derived from these solutions simply by plugging in $\theta_D=\theta_G=0$. Like the solutions in the prior section, these generalized solutions require checking a number of candidates, but this number has grown to account for the possible responses to regulation.

\begin{proposition}
    Given a two-attribute fine-tuning game with quadratic costs, regulatory constraints $\theta_G, \theta_D$,
    and bargaining parameter $\delta$, the domain specialist $D$'s subgame perfect equilibrium strategy is one of the values in the following set:
    $$\gamma_1^* \in \left\{\begin{array}{cc}
        \gamma_0 + \frac{(1-\delta)}{2}C_1^{-1} r, 
    \left[\begin{array}{c}
         \alpha_0 \\
         \beta_0+\frac{(1-\delta)r_\beta}{2c_{1\beta\beta}}
    \end{array}\right],  
         \left[\begin{array}{c}
         \alpha_0 +\frac{(1-\delta)r_\alpha}{2c_{1\alpha\alpha}} - \frac{c_{1\alpha\beta}}{c_{1\alpha\alpha}}\max(0,\theta_D - \beta_0)\\
         \max(\beta_0,\theta_D)
    \end{array}\right], \\
    \left[\begin{array}{c}
         \alpha_0 \\
         \max(\beta_0,\theta_D)
    \end{array}\right],
    \texttt{abstain}.
    \end{array}    
    \right\}$$
The strategy is the feasible candidate which maximizes $U_D$, subject to $U_D\geq 0, \alpha_1 \geq \alpha_0, \beta_1\geq \max(\beta_0,\theta_D)$.
\label{prop:D-strategy-regulation}
\end{proposition}

\begin{proposition}
    Given a two-attribute, two-player fine-tuning game with quadratic costs, regulatory constraints $\theta_G, \theta_D$, and bargaining parameter $\delta$, $G$'s best-response is one of the following candidates:
    \begin{itemize}
        \item $\frac{\delta}{2}C_0^{-1} r, $
        \item $\left[\begin{array}{c}
         0 \\
         \frac{\delta r_\beta}{2c_{0\beta\beta}}
    \end{array}\right], $
    \item $\left[\begin{array}{c}
         \frac{\delta r_\alpha}{2c_{0\alpha\alpha}}-\frac{c_{0\alpha\beta}}{c_{0\alpha\alpha}}\theta_G\\
         \theta_G
    \end{array}\right], $
    \item $\left[\begin{array}{c}
         0\\
         \theta_G
    \end{array}\right], $
    \item \texttt{abstain},
    \item Three additional candidates along the $U_D=0$ constraint, which is given by the following quadratic equation:
    \begin{eqnarray*}
        (1-\delta)r_\alpha \alpha_0 + \left(\frac{(1-\delta)^2r_\alpha^2}{4c_{1\alpha\alpha}} + (1-r_\beta \theta_D - \frac{c_{1\alpha\beta}}{c_{1\alpha\alpha}}(1-\delta)\theta_D + \frac{c_{1\alpha\beta}^2}{c_{1\alpha\alpha}\theta_D^2}-c_{1\beta\beta}\theta_D^2\right) + \\ \left(\frac{c_{1\alpha\beta}}{c_{1\alpha\alpha}}(1-\delta)r_\alpha -2\frac{c_{1\alpha\beta}^2}{c_{1\alpha\alpha}}\theta_D+2c_{1\beta\beta}\theta_D\right)\beta_0 + \left(\frac{c_{1\alpha\beta}^2}{c_{1\alpha\alpha}}-c_{1\beta\beta}\right)\beta_0^2=0.
    \end{eqnarray*}

        \end{itemize}
The strategy is the candidate which maximizes $U_G$, subject to $U_G\geq0, U_D\geq 0, \alpha_1 \geq 0, \beta_1\geq \theta_G$.
\label{prop:G-strategy-regulation}
\end{proposition}

The proof of the above propositions is provided in Appendix \ref{app:player-strategies-regulation}. 
We outline the intuition behind the proof as follows: Notice that the optimization is an inequality-constrained quadratic optimization problem. The problem has been set up so no solutions exist at infinity, that is, the solutions will either be local maxima or will reside on constraints. Therefore, we can find the critical points for the unconstrained problem, as well as the critical points for every possible combination of every constraint in our problem. This yields a set of candidates, which are worked out and listed in the set above.

There is a bit of additional subtlety in the process for arriving at the last three candidates along the constraint listed at the end of the Proposition. Two of the three candidates reside at the intersection of this constraint with the other constraints---that is, they satisfy the constraint listed and either $\alpha_0=0$ or $\beta_0=\theta_G$. Finding the point that satisfies these combinations of constraints is only as hard as solving the roots of a one-variable quadratic, at worst. The third one, however, is a bit more convoluted. This candidate can be described as the solution to the optimization problem $\max_{\gamma_0}U_G \ \ s.t. U_D=0$, where the other constraints are ignored. Although this is a (not necessarily convex) quadratic program, specifying the Lagrangian suggests that its solution must be the solution of a system of three distinct equations with three unknown variables $(\alpha_0,\beta_0,\lambda)\in \mathbb{R}^3$. Two of these equations are quadratic, and the other is linear: 
        \begin{itemize}
            \item $\delta r_\alpha - 2c_{0\alpha\alpha} \alpha_0 - 2c_{0\alpha\beta} \beta_0 - \lambda (1-\delta)r_\alpha = 0$,
            \item $\frac{\delta c_{1\alpha\beta}r_{\alpha}}{c_{1\alpha\alpha}}-2c_{0\beta\beta}\beta_0-2c_{0\alpha\beta}\alpha_0 - \lambda\left(\frac{c_{1\alpha\beta}}{c_{1\alpha\alpha}}(1-\delta)r_\alpha-2\frac{c_{1\alpha\beta}^2\theta_D}{c_{1\alpha\alpha}}+2\left(\frac{c_{1\alpha\beta}}{c_{1\alpha\alpha}}-c_{1\beta\beta}\right)\beta_0\right)=0$,
            \item The quadratic stated in the proposition.
        \end{itemize}
        Though there may be multiple roots satisfying the above equations, the roots are bounded in typical fashion by Bezout's Theorem. Further algebra for arriving at solutions is left to the computer.
\section{Computational results}

Here we describe a set of numerical tests and demonstrations to explore the strategies in our game, using the solved strategies from the previous section. 
Our analysis here is focused on the existence of a persistent facet of the model concerning the way the players shift their strategies in response to regulation. 
With the knowledge that one player or the other is required to meet a regulatory floor, agents can choose their strategies accordingly. In a variety of cases, we observe that the strategies shift in a way that \textit{lowers} the ultimate safety investment compared to safety attained under no regulation. This effect -- which we term \textit{backfiring} -- is observable in cases where the regulation is weak, meaning it imposes a floor that the players already meet under no regulation. 

This section starts by demonstrating the existence of this effect. We then discuss its persistence in cases where players can flexibly choose how they share revenue via a linear contract. Finally, in stark contrast to the observation that regulation can backfire, we find that regulation can act as a commitment device, unlocking strategy sequences that mutually benefit the players. 

\subsection{Regulation can backfire.}

\begin{figure}
    \centering
    \includegraphics[width=0.5\linewidth]{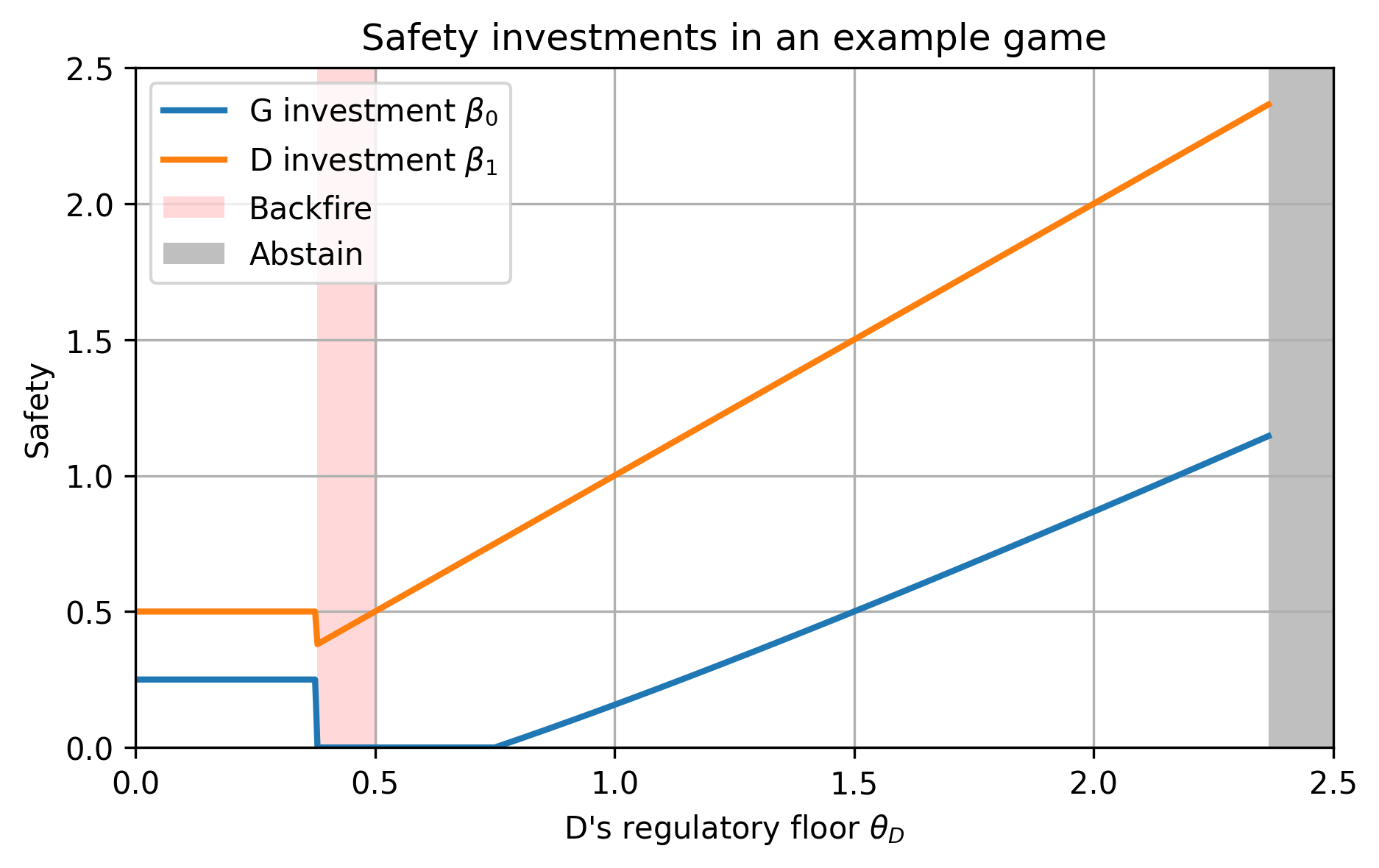}
    \caption{Backfiring observed in a basic two-player game where $\theta_G=0$ and $\theta_D$ is varied over the range $[0,2.5]$. As $\theta_D$ is swept upward from $0$, there is some value at which the generalist's score exhibits a discontinuity and the investment in safety lowers. In this example, the discontinuity occurs at a threshold value below the safety attained under no regulation ($0.5$). In response to this discontinuity in the generalist's strategy, the domain-specialist minimally complies with the regulation, meaning the ultimate safety is reduced for some non-zero regulations. This plot conveys information redundant with the first panel of Figure \ref{fig:fourpanel-separable}, where the regulation is swept only along the vertical line in which the generalist's threshold is $0$.}
    \label{fig:backfiring-examples.}
\end{figure}

Consider a basic game given by the following set of cost and revenue parameters: $C_1 = C_0 = I_2, r_\alpha=r_\beta=1, \delta=0.5$. 
This game is \textit{separable}, meaning there are no interaction effects between performance and safety, and it assumes the market without regulation places equal value on performance and safety. 
Figure \ref{fig:backfiring-examples.} depicts the players' strategies in this game, for varying levels of regulation targeting the Domain-specialist alone. For the lowest regulatory thresholds, we observe that the players stick to their no-regulation safety investments, since they already clear the threshold and their no-regulation investments remain optimal. As the regulatory floor is increased, however, the generalist's strategy exhibits a discontinuity. Crucially, this drop in G's safety investment occurs at a regulatory threshold \textit{lower than} the no-regulation safety strategy. 

Why does $G$ switch strategies? Here we attempt to provide some intuition. Abent regulation, there exists some typical best-response that $D$ will take, and $G$ must anticipate this best-response to choose an optimal strategy. Even if $G$ would theoretically prefer $D$ to invest more in safety, $G$ cannot fully control $D$'s actions. Regulation that targets $D$, however, does just this: it restricts $D$'s strategy space so $D$ must commit to certain safety investments, regardless of $G$'s strategy. 
Therefore, in the presence of regulation, $G$ is incentivized to select a new strategy sequence with lower investments in safety, because $G$ knows that $D$ will cover the gap in safety between $G$'s investment and $D$'s threshold. Put another way, $G$ is given the opportunity to engage in a kind of \textit{free riding} behavior. $G$ creates a gap in the safety investment as a cost-cutting exercise, knowing $D$ must bridge the gap to reap any reward in the game. 

 Our results here convey that backfiring occurs in one instance of the game. We have yet to give a clear characterization of how widespread this phenomenon is.
 The example we have shown so far assumes the revenue-sharing parameter is fixed at $\delta=0.5$. One might imagine that, instead of a fixed revenue sharing arrangement, players can collectively decide how to share the revenue. Does the ability to influence the revenue-sharing parameter prevent cases where weak regulation backfires? We turn to this question next.

\subsection{Bargaining does not suffice to prevent backfiring.}

\begin{figure}
    \centering
    \includegraphics[width=\linewidth]{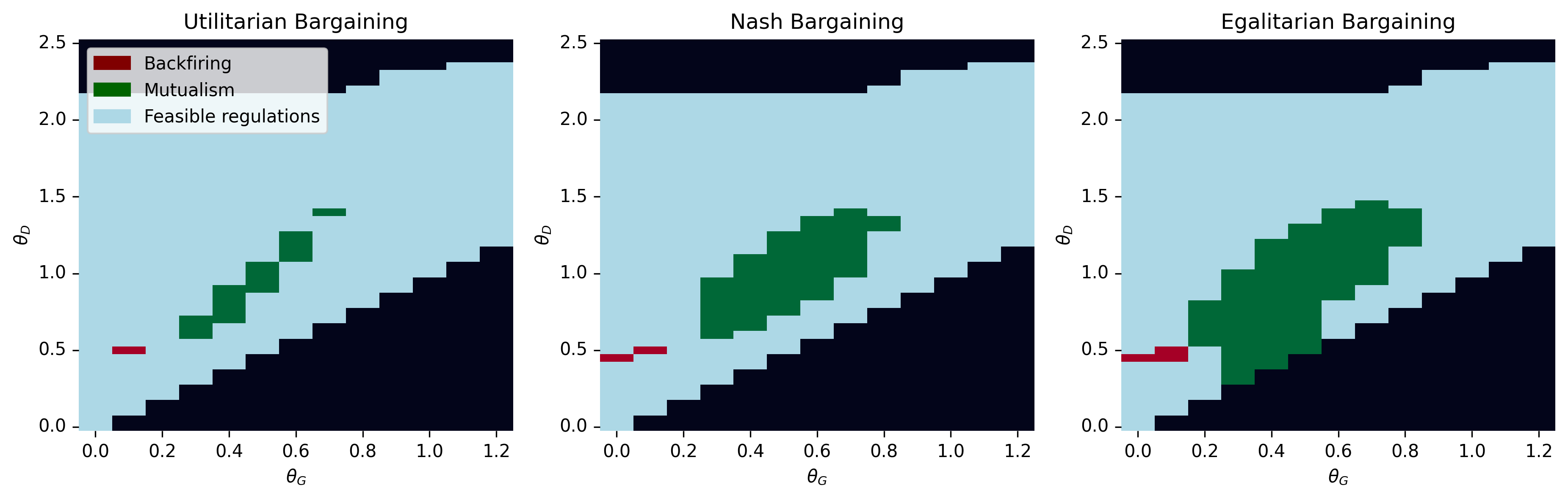}
    \caption{Results from numerical tests over the set of possible $(\theta_G,\theta_D)$ pairs in the two-attribute, two-player, separable quadratic-cost game. %Backfiring is detected in all three cases. 
    Backfiring occurs in the weak regulatory regimes in which $\theta_D$ is just below $\beta_0^A$.  
    Regulations that mutually improve both players' utilities over anarchy are detected for all three bargaining solutions. The highest aggregate utility in this game is achieved at $\theta_G=0.5,\theta_D=1$.}
    \label{fig:bargaining}
\end{figure}

Here we provide evidence that the existence of backfiring persists in more games beyond the example portrayed in Figures \ref{fig:fourpanel-separable}. In particular, we relax the assumption that players share their revenue according to a constant revenue-sharing parameter $\delta=0.5$. Instead, we allow players to reach \textit{bargaining agreements} to distribute revenue --- and, correspondingly, profit --- in a way that maximizes their joint utility. Bargaining solutions are arrangements that maximize the players' joint utility.\footnote{The relevance of bargaining solutions to our setting is described in further depth by \citet{laufer2024fine}.} We provide evidence that even when players can distribute revenue in a way that maximizes the joint utility, these arrangements can still exhibit backfiring effects. We assume here that the players jointly agree on a bargaining solution \textit{before} either invests effort, but \textit{after} learning about the regulation.\footnote{The next sections will relax this assumption further, providing findings on the existence of backfiring and mutualism for every non-trivial linear revenue-sharing agreement $\delta \in (0,1)$.} Figure \ref{fig:bargaining} shows the numerical results for a variant of the separable game where we vary the value of $\delta$ over 98 values in the range $[0.01,0.99]$. We vary the regulatory setting for 13 $\theta_G$ values in $[0,1.2]$ and 51 $\theta_D$ values in $[0,2.5]$, for a total of 49,686 simulated games. The figure depicts three different processes for arriving at an optimal bargain: \textit{utilitarian}, which selects $\delta$ to maximize the sum of utilities, \textit{Nash}, which selects $\delta$ to maximize the product of utilities \cite{nash1950bargaining}, and \textit{egalitarian}, which sets $\delta$ to maximize the minimum of the utilities. In all scenarios, we observe at least one instance of a combination of regulations that backfire. Further, we observe a cluster of regulation regimes that yield mutual improvement to utility. The results suggest that even in instances where players can choose how to distribute revenue through revenue-sharing, these agreements are sub-optimal for engendering the right sort of commitment from each of the players, if their goal is to mutually benefit from their interaction. In the next section, we explore the idea that regulation can bring about a \textit{mutualistic} benefit, beyond what is achievable through bargaining. 

\subsection{Regulation can act as a commitment device.}
\label{subsec:commitment}

\begin{figure}
    \centering
    \includegraphics[width=0.4\linewidth]{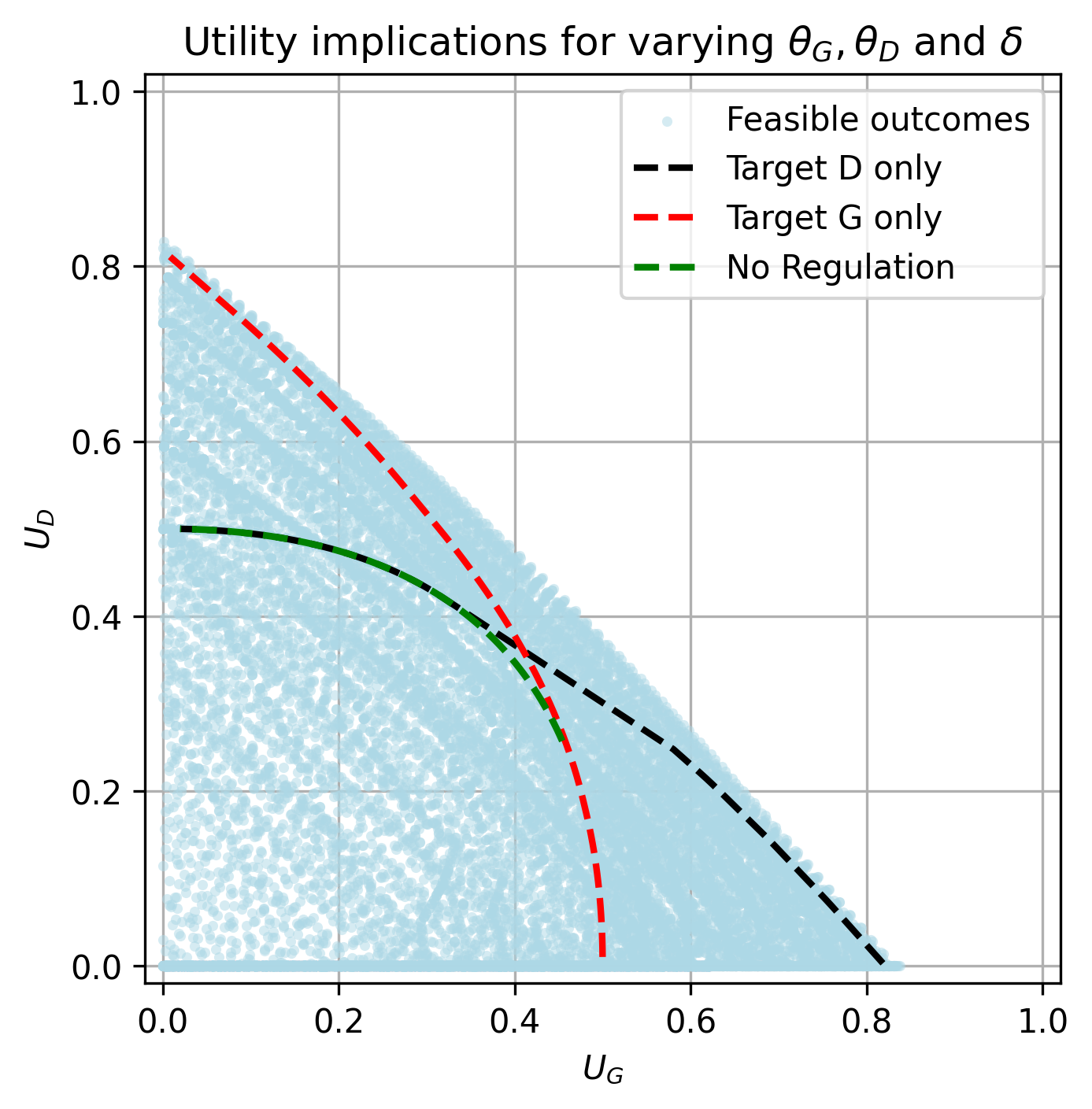}
    \caption{The set of attainable utility outcomes over a grid of possible regulation regimes and bargains for the two-attribute, two-player, separable quadratic-cost game. If we imagine that regulations are \textit{endogenous} to the game -- that is, regulations are decided collectively by the players, like $\delta$, then each of the blue points represents a possible game with utility implications for the two players. If the players are restricted to a particular regulatory regime -- targeting G only, D only, or neither -- then the utility they are able to achieve (depicted in dashed lines) suffers, compared to the regime where both players are subjected to regulation. The shape formed by the light blue dots represents the full set 
    of simulated games, suggesting that regulation targeting both players is at times needed to achieve outcomes that cannot be achieved by regulating just one.}
    \label{fig:pareto}
\end{figure}

Here we show that there exist cases where regulation can leave both players \textit{better off} than anarchy, while also benefiting the safety of the technology. Even though the regulation constrains the space of investments that players are able to achieve, it can nonetheless leave each player with higher utility than they are able to achieve under no regulation. 
To make this finding more clear, we depict the set of all achievable $(U_G, U_D)$ combinations in Figure \ref{fig:pareto}. The light blue cloud of points represents all attainable utility scenarios, over a grid of $\theta_D, \theta_G,$ and $\delta$ values. The dotted lines represent the convex hull (northeastern faces) of attainable utility implications for the following regimes: 1) neither player is targeted with regulation (depicted in green), 2) one player is targeted with regulation (depicted in red and black), and 3) both players are targeted with regulation (inferrable from the outermost feasible points). The figure suggests that a non-vacuous constraint on \textit{both} players achieves more preferable utility outcomes than regulations of individual players or bargaining alone are able to achieve.

These results suggest that, although regulation can backfire, it can also mutually serve the interests of both players while also improving the level of safety of the technology. This finding raises the following question: if it was possible to achieve higher utilities all around, why was this set of strategies not chosen by the players in the unregulated game? The players did not opt for this set of strategies because these strategies are dominated for at least one player in at least one subgame. As a hypothetical, imagine that under no regulation, the players sit down for a conversation before the game, and both say they will contribute $\epsilon$ additional investment in safety. When the game reaches the last step, however, $D$ finds he benefits \textit{more} from investing only $\frac{(1-\delta)}{2} C_1^{-1}$, rather than the agreed upon value of $\frac{(1-\delta)}{2} C_1^{-1} + 2\epsilon$. What's more, $G$ \textit{knows} that $D$ will do this, and so $G$'s decision will break the agreement before $D$ even gets the chance to respond. Without the regulation restricting $D$'s behavior away from changing strategies in the final hour, nothing prevents $D$ from pursuing the highest-utility strategy, even if it harms $G$. Thus, our model has a prisoner's dilemma dynamic baked into it: there are feasible strategies that leave both players better off, but these strategies are not equilibria.

Absent regulation, the players might \textit{wish} they could ensure the other will uphold their side of a verbal agreement, though they are unable to guarantee it. Regulation, therefore, can act as a \textit{commitment device}, which lends teeth to agreements that the players are able to enter prior to making their investments. This commitment device can be valuable in a formal sense: Both players would be willing to pay for it, as long as the price is less than the amount of utility they collectively gain under regulation.

\section{A General Characterization} 

In the previous sections, we arrived at closed-form solutions for the players' strategies and have demonstrated individual instances that exhibit the backfiring effect of regulation. 
We have not yet determined how widespread this phenomenon is. In this section, we provide analytical results that characterize when this phenomenon occurs. Our findings suggest that this effect is notably widespread. 
We find that for all quadratic-cost games, backfiring occurs as long as both of the technology's attributes (performance and safety) are sufficiently \textit{complementary} such that, under no regulation, the players will invest in some combination of them. Intuitively, if the players invested only in performance under no regulation, backfiring would be impossible as the baseline safety investment would be zero. Therefore, our condition for backfiring covers all games where the market prefers some non-zero baseline investment in performance and safety. The condition we rely on is precisely the condition introduced in Remark \ref{remark:unconstrained-condition}, which represents an upper bound on the cost interaction terms. This section will prove that both backfiring and mutualism occur in a range of scenarios that depend crucially on the cost interaction term, and will describe what this dependence looks like. 

\subsection{Backfiring occurs in all mixed-strategy games.}

Below we prove that for all AI regulation games in which the players invest a non-zero amount in safety and performance under no regulation, there is a non-empty set of regulatory regimes that exhibit a backfiring effect.

\begin{theorem}
    Given an AI regulation game with quadratic costs. If both players' cost interactions meet the following conditions: 
    $$c_{p,\alpha\beta} < \min\left(\sqrt{c_{p,\alpha\alpha}c_{p,\beta\beta}},\frac{c_{p,\alpha\alpha} r_\beta}{r_\alpha}, \frac{c_{p,\beta\beta}r_\alpha}{r_\beta}\right),$$
    %conditions in Remark \ref{remark:unconstrained-condition}, 
    then there exists an $\epsilon>0$ such that the regulatory regime $\theta_G = 0, \theta_D =  \beta_0^A - \epsilon$ backfires. 
\label{thm:backfiring}
\end{theorem}

The proof of the above theorem is provided in Appendix \ref{app:backfiring-proof}. Here we provide an overview of the conceptual argument. We start by observing that the unregulated optimal strategies $\gamma_0^A,\gamma_1^A$ remain feasible in weak regulatory settings. These strategies dominate all alternative strategies in which the players contribute to safety \textit{beyond} their regulatory constraints, as any such strategy was available in the no regulation scenario, so they were already shown to be sub-optimal compared to $\gamma_0^A,\gamma_1^A$. The proof's task, therefore, is is to find some $\theta_D<\beta_1^A$ and some $\gamma_0' \neq \gamma_0^A$, such that $D$ \textit{minimally complies} with the regulation ($\beta_0' = \theta_D)$, and further, $U_G(\gamma_0'; \theta_D) > U_G(\gamma_0^A;\theta_D)$. For the proof to work, we choose a regulation of $\theta_G = 0, \theta_D=\beta_0^A -\epsilon$ for some small positive $\epsilon>0$, and generalist strategy $\gamma_0' = \left[\begin{array}{c}
     \frac{\delta r_\alpha}{2c_{0,\alpha\alpha}}\left(\beta_0^A - 2\epsilon\right) \\
        \beta_0^A-2\epsilon
\end{array}\right]$. For sufficiently small $\epsilon$, we find that the change to the utility of $G$ for using this strategy is positive as long as the following condition is met: $r_\beta > \frac{c_{1,\alpha\beta}}{c_{1,\alpha\alpha}}r_\alpha$. This inequality, given by the analysis in Appendix \ref{app:backfiring-proof}, is precisely the condition established in Remark \ref{remark:unconstrained-condition} for non-zero investment in safety under no regulation. 

The above results demonstrate that backfiring does not only exist in single degenerate cases: It occurs in a range of scenarios in which players share revenue and each contribute non-zero effort to the development of the technology. These scenarios include settings in which the two attributes are \textit{complementary}, as well as a range of settings where the two attributes are \textit{interfering}, up to a particular limit that we are able to specify. We note that further generalizations are open for broader functional forms, including more expressive polynomial costs and exponential costs. The generality of the backfiring effect in the quadratic case gives us reason to believe that the effect might hold for a broader set of forms, though we leave these directions to future work. 

\subsection{Mutualism occurs in sufficiently separable games.}
\label{subsec:mutulaism-proof}

So far, we have shown that a set of regulations backfire in a swath of two-attribute games. Here we provide a second result on a set of regulations that fare better. Using similar logic about games with bounded interaction effects between the the two attributes, we find that there exist combinations of regulatory thresholds that mutually improve the two players' utilities, as well as the safety level of the technology. 
We state this result below. 

\begin{theorem}
    Given a two-player AI regulation game with quadratic costs. If both players meet the following conditions:
    $$|c_{p,\alpha\beta}| < \min\left(\sqrt{c_{p,\alpha\alpha}c_{p,\beta\beta}},\frac{c_{p,\alpha\alpha} r_\beta}{r_\alpha}, \frac{c_{p,\beta\beta}r_\alpha}{r_\beta}\right),$$
    then there exists an $\epsilon>0$ such that the regulatory regime $\theta_G = \beta_0^A + \epsilon, \theta_D = \beta_1^A + 2\epsilon$ mutually improves both players' utilities.
\label{thm:mutualism}
\end{theorem}

The proof of the above theorem is given in Appendix \ref{app:mutualism}. The proof follows a similar strategy to the backfiring proof. We are focused on the set of games where the players arrive at unconstrained solutions in the case of no regulation, and we perturb the regulation by a small positive $\epsilon$ value and see the implications for the players' utilities. Here, instead of targeting only the domain-specialist and specifying a threshold slightly below the unconstrained optimal strategy, we set the regulation to target \textit{both players} using a threshold slightly \textit{above} their unconstrained strategies. Instead of measuring the impact on safety, we measure the impact on the players' utilities and find that, under the specified condition, the utilities both improve.  

The results suggest that, similar to the characterization of backfiring, the mutualism effect is observable in a range of quadratic-cost games, including in \textit{separable scenarios} and a range of \textit{complementary} and \textit{interfering} scenarios. Notice, however, that our condition for establishing when mutualism occurs is slightly different than the condition in the backfiring theorem. Instead of a one-sided bound on the players' cost interaction terms, our proof relies on a two-sided bound. The analysis suggests there may be certain games where the two attributes are \textit{strongly complementary} where slightly increasing the regulation in the manner proposed does not increase players' utilities. In other words, if the market already sufficiently incentivizes joint investments in safety and performance, then forcing safety requirements on both players in equal proportion may not benefit players' utilities. In these cases, a linear contract may suffice to serve the utilities of the players, and so regulation would only be needed for achieving the goal of advancing safety, and would not serve the additional role in enforcing commitments from players. 

\section{Conclusion}

Proposals for AI regulation have made use of the idea that different entities contribute to these technologies in succession. This work provides a model for reasoning about the effects of targeting AI safety regulation along the development chain. Our findings suggest that \textit{weak} safety regulation predominantly targeted at the domain specialist can backfire, yielding lower investments in safety than in the alternative case of no regulation. Our findings further suggest that regulation appropriately targeted at both upstream producers and downstream specialists can exhibit a mutualism effect in which both entities benefit. After demonstrating instances of the backfiring and mutualism effects through a numerical simulation, we provide analysis showing these phenomena are not just degenerate cases but hold in a range of scenarios. 

Our results reveal natural directions for future research. In the setting we have put forward, it would be interesting to move beyond showing the existence of backfiring and mutualism regions and characterize the shape of these regions and the magnitude of their effects. Certain segments of the boundaries of these regions are straightforward but others seem to require solving higher-order polynomials to express in closed-form. 

Generalizations beyond the quadratic-cost games might be interesting. For instance, it may be possible to show that backfiring and Pareto-improvement effects occur for any convex cost and concave revenue games meeting where there exist some marginal conditions on the functions' marginal conditions including their slopes and intercepts. 

We have predominantly focused on the case where there is one domain-specialist, but in many real-world settings the development of AI technologies involve multiple domains, and each domain may involve many entities who compete. To what extent does competition between multiple entities change the backfiring and Pareto-improving impacts of regulation? Pursuing questions about multiple domain-specialists would require further specifying the structure of $G$'s contract with each specialist, which might reasonably be conceived as a constant revenue share across domains, a constant fixed price across domains, or a variable price across domains. Relatedly, approaches to regulating different specialists may be conceived of as domain-specific (different requirements for each domain) or domain-agnostic (requirements for all domains). Pursuing questions about multiple generalists may also illuminate interesting directions. In particular, if different domains have different preferences over attributes, there may be scenarios where general providers \textit{specialize} their investments to capture some domains and cede others to their competitors. Such dynamics raise new questions about how to design regulation to account for these rich constellations of interacting actors.

\begin{acks}
    
\end{acks}

\bibliography{bibliography}

\section{Game Solving}

\subsection{Player's strategies without regulation}
\label{app:noregulation}

\textbf{The domain-specialist's strategy.} The proof for Proposition \ref{prop:D-strategy-noreg} is given below.
\begin{proof}
    $D$'s best-response strategy is the value $\gamma_1^*$ that maximizes $D$'s utility. 
    \begin{eqnarray*}
    \gamma_1^*(\gamma_0, \delta) &= & \arg\max_{\gamma_1}U_D(\gamma_0, \gamma_1 \delta) \text{ s.t. } U_D\geq 0, \alpha_1 \geq \alpha_0, \beta_1 \geq \beta_0\\
    \end{eqnarray*}
    Observe that $D$ will not abstain because zero-investment ($\gamma_1=\gamma_0$) is cost-free, yielding non-negative utility, so we can safely ignore the constraint.
    To solve the optimization, we specify the Lagrangian as follows for some multipliers $\lambda_1 \in \mathbb{R}, \lambda_2 \in \mathbb{R}$ and a slack variables $s_1 \in \mathbb{R}, s_2 \in \mathbb{R}$. By construction, we assert that the slack variables are only non-zero when the multipliers are zero, and the multipliers are non-zero only if the slack variables are zero.
    $$\mathcal{L}  :=   (1-\delta) r^T \gamma_1 - (\gamma_1-\gamma_0)^T C_1 (\gamma_1-\gamma_0) - \lambda_1(\alpha_1-\alpha_0-s_1^2) - \lambda_2 (\alpha_1-\alpha_0-s_2^2). $$
    We partially differentiate with respect to each decision variable and each multiplier.
    \begin{eqnarray*}
    \frac{\partial}{\partial \alpha_1} \mathcal{L} &=& 0\\
     \iff (1-\delta) r_\alpha - 2c_{1,\alpha\alpha} (\alpha_1-\alpha_0) + 2c_{1,\alpha\beta} (\beta_1 - \beta_0)- \lambda_1 &=& 0 \\
    \frac{\partial}{\partial \beta_1} \mathcal{L} &=& 0\\
     \iff (1-\delta) r_\beta - 2c_{1,\beta\beta} (\beta_1 - \beta_0) + 2c_{1,\alpha\beta} (\alpha_1-\alpha_0)- \lambda_2 &=& 0 \\
    \frac{\partial}{\partial \lambda_1} \mathcal{L} &=& 0 \\
    \iff -\alpha_1 + \alpha_0 + s_1^2  &=& 0 \\
    \frac{\partial}{\partial \lambda_2} \mathcal{L} &=& 0 \\
    \iff -\beta_1 + \beta_0 + s_2^2  &=& 0 \\
\end{eqnarray*}
Using complementary slackness, we have four possible options:
\begin{enumerate}
    \item $s_1=0, \lambda_1 > 0, s_2 = 0, \lambda_2>0 \rightarrow \beta_1^*=\beta_0, \alpha_1^*=\alpha_0.$
    \item $s_1\neq0, \lambda_1 = 0, s_2 = 0, \lambda_2>0 \rightarrow \beta_1^*=\beta_0$, and we can plug into our first of four equations above:
    \begin{eqnarray*}
    (1-\delta) r_\alpha - 2c_{1,\alpha\alpha} (\alpha_1-\alpha_0) + 2c_{1,\alpha\beta} (\beta_1 - \beta_0) -\lambda_1&=& 0 \\     
    \rightarrow  (1-\delta) r_\alpha - 2c_{1,\alpha\alpha} (\alpha_1-\alpha_0) &=& 0 \\     
    \rightarrow \alpha_1^* = \alpha_0 + \frac{(1-\delta)r_\alpha}{2c_{1,\alpha\alpha}}.
    \end{eqnarray*}
    \item $s_1=0, \lambda_1 > 0, s_2 \neq 0, \lambda_2=0 \rightarrow \alpha_1^*=\alpha_0,$ and we can plug in to equation 2:
    \begin{eqnarray*}
        (1-\delta) r_\beta - 2c_{1,\beta\beta} (\beta_1 - \beta_0) + 2c_{1,\alpha\beta} (\alpha_1-\alpha_0)- \lambda_2 &=& 0 \\
        \rightarrow (1-\delta) r_\beta - 2c_{1,\beta\beta} (\beta_1 - \beta_0) - \lambda_2 &=& 0\\
        \rightarrow \beta_1^* = \beta_0 + \frac{(1-\delta)r_\beta}{2c_{1,\beta\beta}}.
    \end{eqnarray*}
    \item $s_1\neq0, \lambda_1 = 0, s_2 \neq 0, \lambda_2=0 \rightarrow $ This is the unconstrained critical point, and is solved via the first two systems of equations:
    \begin{eqnarray*}
        \nabla U_D = (1-\delta)r - 2C_1(\gamma_1-\gamma_0) &=& 0\\
        \rightarrow \gamma_1^* = \gamma_0 + \frac{(1-\delta)}{2}C_1^{-1} r.
    \end{eqnarray*}
\end{enumerate}
Thus we have established our four candidates in the proposition statement.
\end{proof}

\textbf{The Generalist's strategy.} The proof for Proposition \ref{G-strategy-noreg} is given below.
\begin{proof}
    $G$'s best-response strategy is the value $\gamma_0^*$ that maximizes $G$'s utility.
    \begin{eqnarray*}
    \gamma_0^*(\delta) &= & \arg\max_{\gamma_1}U_G(\gamma_0, \delta) \text{ s.t. } U_G \geq 0, \alpha_0 \geq 0, \beta_0 \geq 0.\\
    \end{eqnarray*}
    Following the same steps as the proof of Proposition \ref{prop:D-strategy-noreg}, we specify the Lagrangian as follows for multipliers $\lambda_1 \in \mathbb{R}, \lambda_2 \in \mathbb{R}$ and a slack variables $s_1 \in \mathbb{R}, s_2 \in \mathbb{R}$. 
    $$\mathcal{L}  :=   \delta r^T \gamma_1 - \gamma_0^T C_0 \gamma_0 - \lambda_1(\alpha_0-s_1^2) - \lambda_2 (\alpha_0-s_2^2). $$
    We partially differentiate with respect to each decision variable and each multiplier.
    \begin{eqnarray*}
    \frac{\partial}{\partial \alpha_0} \mathcal{L} &=& 0\\
     \iff \delta r_\alpha - 2c_{0,\alpha\alpha} \alpha_0 + 2c_{0,\alpha\beta} \beta_0- \lambda_1 &=& 0, \\
    \frac{\partial}{\partial \beta_1} \mathcal{L} &=& 0\\
     \iff \delta r_\beta - 2c_{0,\beta\beta}  \beta_0 + 2c_{1,\alpha\beta} \alpha_0- \lambda_2 &=& 0, \\
    \frac{\partial}{\partial \lambda_1} \mathcal{L} &=& 0 \\
    \iff -\alpha_0 + s_1^2  &=& 0, \\
    \frac{\partial}{\partial \lambda_2} \mathcal{L} &=& 0 \\
    \iff -\beta_0 + s_2^2  &=& 0. \\
\end{eqnarray*}
Using complementary slackness, we have four possible options:
\begin{enumerate}
    \item $s_1=0, \lambda_1 > 0, s_2 = 0, \lambda_2>0 \rightarrow \beta_0^*=0, \alpha_0^*=0.$
    \item $s_1\neq0, \lambda_1 = 0, s_2 = 0, \lambda_2>0 \rightarrow \beta_0^*=0$, and we can plug into our first of four equations above:
    \begin{eqnarray*}
    \delta r_\alpha - 2c_{0,\alpha\alpha} \alpha_0 + 2c_{0,\alpha\beta}  \beta_0 -\lambda_1&=& 0 \\     
    \rightarrow  \delta r_\alpha - 2c_{0,\alpha\alpha} \alpha_0 &=& 0 \\     
    \rightarrow \alpha_0^* = \frac{\delta r_\alpha}{2c_{0,\alpha\alpha}}.
    \end{eqnarray*}
    \item $s_1=0, \lambda_1 > 0, s_2 \neq 0, \lambda_2=0 \rightarrow \alpha_0^*=0,$ and we can plug in to equation 2:
    \begin{eqnarray*}
        \delta r_\beta - 2c_{0,\beta\beta}  \beta_0 + 2c_{0,\alpha\beta} \alpha_0- \lambda_2 &=& 0 \\
        \delta r_\beta - 2c_{0,\beta\beta}  \beta_0 &=& 0\\
        \rightarrow \beta_0^*=  \frac{\delta r_\beta}{2c_{0,\beta\beta}}.
    \end{eqnarray*}
    \item $s_1\neq0, \lambda_1 = 0, s_2 \neq 0, \lambda_2=0 \rightarrow $ This is the unconstrained critical point, and is solved via the first two systems of equations:
    \begin{eqnarray*}
        \nabla U_G = \delta r - 2C_0 \gamma_0 &=& 0\\
        \rightarrow \gamma_0^* =  \frac{\delta}{2}C_0^{-1} r.
    \end{eqnarray*}
\end{enumerate}
Thus we have established our four candidates.   
\end{proof}

\subsection{Condition for non-zero performance and safety investment}
\label{app:non-zero-investment-no-reg}

\textbf{Condition establishing non-zero investment}. Below we prove Remark \ref{remark:unconstrained-condition}.
\begin{proof}
    The first of the three inequalities establishes that the player's costs are strictly convex: 
    \begin{eqnarray*}
        c_{\alpha\beta}< \sqrt{c_{p,\alpha\alpha}c_{p,\beta\beta}} \iff c_{p,\alpha\alpha}c_{p,\beta\beta}-c_{\alpha\beta}^2>0 \iff \det C_p>0.
    \end{eqnarray*}
    By the spectral theorem, we know a 2x2 matrix is positive definite if and only if its determinant and trace are both positive, which is now established.
    By Lemma \ref{lemma:pd-iff-concave}, the utility is strictly concave for our setting if and only if the cost is strictly convex. Thus the unconstrained solution is the global optimum as long as it is feasible. Thus, the necessary and sufficient condition for optimality is the condition for feasibility. 
    \begin{itemize}
        \item 
    For the generalist: $$\frac{\delta}{2}C_0^{-1}r > 0 \iff \frac{\delta}{2\det C_0}\left[\begin{array}{c}
    c_{0,\beta\beta}r_\alpha-c_{0,\alpha\beta}r_\beta \\
          -c_{0,\alpha\beta}r_\alpha+c_{0,\alpha\alpha}r_\beta
    \end{array}\right]>\left[\begin{array}{c}
         0 \\
         0
    \end{array}\right]$$
    Using the same positive definiteness identity above, we know the determinant is positive. We are given $\delta>0$. 
    Thus we can cancel the positive constant term $\frac{\delta}{2\det C_0}$. The two inequalities simplify to those stated in the proposition.
    \item For the specialist, the proof proceeds identically. Observe that $(1-\delta)\geq 0$ and the unconstrained contribution is given by:
    $\frac{1-\delta}{2} C_1^{-1}r$. 
    \end{itemize}
\end{proof}

\subsection{Proof for Player Strategies with Regulation}
\label{app:player-strategies-regulation}

Here we provide proofs for our propositions establishing best-response strategies for the players. 

\textbf{Domain-specialist best-response under regulation.} Here we provide the proof of Proposition \ref{prop:D-strategy-regulation}, the domain specialist's best response under regulatory requirement $\theta_D$.

\begin{proof}
    $D$'s best-response strategy is the value $\gamma_1^*$ that maximizes $D$'s utility. $D$ will abstain if and only if the best option yields negative utility.
    \begin{eqnarray*}
    \gamma_1^*(\gamma_0, \delta, \theta_D) &= & \arg\max_{\gamma_1}U_D(\gamma_0, \delta, \theta_D) \text{ s.t. } U_D \geq 0, \alpha_1 \geq \alpha_0, \beta_1 \geq \max\left(\beta_0,\theta_D\right).\\
    \end{eqnarray*}
    Define $\kappa=\max(\beta_0,\theta_D).$ To solve the optimization, we specify the Lagrangian as follows for some multipliers $\lambda \in \mathbb{R}^3$ and a slack variables $s\in \mathbb{R}^3$.
    $$\mathcal{L}  :=   (1-\delta) r^T \gamma_1 - (\gamma_1-\gamma_0)^T C_1 (\gamma_1-\gamma_0) - \lambda_1 (\alpha_1-\alpha_0-s_1^2) - \lambda_2(\beta_1 - \kappa - s_2^2) - \lambda_3(U_D-s_3^2). $$
    We partially differentiate with respect to each decision variable and each multiplier.
    \begin{eqnarray*}
    \frac{\partial}{\partial \alpha_1} \mathcal{L} &=& 0\\
     \iff (1-\delta) r_\alpha - 2c_{1,\alpha\alpha} (\alpha_1-\alpha_0) + 2c_{1,\alpha\beta} (\beta_1 - \kappa)- \lambda_1 - \lambda_3 \frac{\partial U_D}{\partial \alpha_1} &=& 0 \\
     \iff (1-\lambda_3)\left((1-\delta) r_\alpha - 2c_{1,\alpha\alpha} (\alpha_1-\alpha_0) + 2c_{1,\alpha\beta} (\beta_1 - \kappa) \right) - \lambda_1&=& 0 \\
    \frac{\partial}{\partial \beta_1} \mathcal{L} &=& 0\\
     \iff (1-\delta) r_\beta - 2c_{1,\beta\beta} (\beta_1 - \kappa) + 2c_{1,\alpha\beta} (\alpha_1-\alpha_0)- \lambda_2 - \lambda_3 \frac{\partial U_D}{\partial \beta_1} &=& 0 \\
     \iff (1-\lambda_3) \left((1-\delta) r_\beta - 2c_{1,\beta\beta} (\beta_1 - \kappa) + 2c_{1,\alpha\beta} (\alpha_1-\alpha_0)\right)- \lambda_2  &=& 0 \\
    \frac{\partial}{\partial \lambda_1} \mathcal{L} &=& 0 \\
    \iff -\alpha_1 + \alpha_0 + s_1^2  &=& 0 \\
    \frac{\partial}{\partial \lambda_2} \mathcal{L} &=& 0 \\
    \iff -\beta_1 + \kappa + s_2^2  &=& 0 \\
    \frac{\partial}{\partial \lambda_3} \mathcal{L} &=& 0 \\
    \iff -U_D + s_2^2  &=& 0 \\
    \iff -(1-\delta) r^T \gamma_1 + \left(\gamma_1-\gamma_0\right)^T C_1 \left(\gamma_1-\gamma_0\right) + s_2^2  &=& 0 \\
\end{eqnarray*}
Using complementary slackness, we have eight possible options:
\begin{enumerate}
    \item $s_1=0, \lambda_1 > 0, s_2 = 0, \lambda_2>0, s_3\neq0, \lambda_3 = 0 \rightarrow \beta_1^*=\kappa, \alpha_1^*=\alpha_0.$
    \item $s_1=0, \lambda_1 > 0, s_2 = 0, \lambda_2>0, s_3=0, \lambda_3 > 0 \rightarrow \beta_1^*=\kappa, \alpha_1^*=\alpha_0.$ This offers the same candidate as (1).
    \item $s_1=0, \lambda_1 > 0, s_2 \neq  0, \lambda_2=0, s_3 \neq 0, \lambda_3 = 0 \rightarrow \alpha_1^* = \alpha_0, $ solve equations (1) and (2) for $\beta_1^*$ and $\lambda_1$. Omitting the algebra, this yields:
    $$\gamma_1^* = \left[\begin{array}{c}
         \alpha_0 \\
         \beta_0+\frac{(1-\delta)r_\beta}{2c_{1\beta\beta}}
    \end{array}\right]$$
    \item $s_1=0, \lambda_1 > 0, s_2 \neq  0, \lambda_2=0, s_3=0, \lambda_3 > 0 \rightarrow \alpha_1^* = \alpha_0, $ solve equations (1) and (2) for $\beta_1^*$ and $\lambda_1$. This solution, if it is distinct from the previous solution (3), will always be dominated because it is characterized by 0 utility for $G$. 
    \item $s_1\neq0, \lambda_1 = 0, s_2 = 0, \lambda_2>0, s_3=0, \lambda_3 > 0 \rightarrow \beta_1^* = \kappa \rightarrow$ solve equations (1) and (2) for $\lambda_1$ and $\alpha_1^*$. Omitting algebra, this yields:
    $$\gamma_1^* = \left[\begin{array}{c}
         \alpha_0 +\frac{(1-\delta)r_\alpha}{2c_{1\alpha\alpha}} - \frac{c_{1\alpha\beta}}{c_{1\alpha\alpha}}\max(0,\theta_D - \beta_0)\\
         \max(\beta_0,\theta_D)
    \end{array}\right]$$
    \item $s_1\neq0, \lambda_1 = 0, s_2 = 0, \lambda_2>0, s_3\neq0, \lambda_3 = 0 \rightarrow $ this solution, if it is distinct from the previous one (5), will always be dominated because it is characterized by 0 utility for $G$. 
    \item $s_1\neq0, \lambda_1 = 0, s_2 \neq  0, \lambda_2=0, s_3 \neq 0, \lambda_3 = 0 \rightarrow \alpha_1^* = \alpha_0, $ solve equations (1) and (2) for $\alpha_1^*, \beta_1^*$. This is the unconstrained solution. Omitting algebra, this yields:
    $\gamma_1^* = \gamma_0+\frac{(1-\delta)}{2}C_1^{-1}r$.
    \item $s_1\neq0, \lambda_1 = 0, s_2 \neq  0, \lambda_2=0, s_3=0, \lambda_3 > 0 \rightarrow \beta_1^*=\kappa \rightarrow$ solve equations (1) and (2) for $\alpha_1^*,\beta_1^*$. This solution, if it is distinct from (7), will always be dominated by (7) because it is characterized by 0 utility for $G$.
\end{enumerate}
Thus we have established our four candidates in the proposition statement.
To handle the \texttt{abstain} scenario, we check each candidate produced in the process above by plugging the strategy to our formula for $U_D$. If none yield positive utility, then the domain specialist prefers to \texttt{abstain}.
\end{proof}

\textbf{The generalist's subgame perfect equilibrium strategy under regulation}. Here we prove Proposition \ref{prop:G-strategy-regulation}. 

\begin{proof}
    $G$'s best-response strategy is the value $\gamma_0^*$ that maximizes $G$'s utility.
    \begin{eqnarray*}
    \gamma_0^*(\delta, \theta_G, \theta_D) &= & \arg\max_{\gamma_0}U_G(\gamma_0; \delta, \theta_G, \theta_D) \text{ s.t. } U_G \geq 0, U_D \geq 0, \alpha_0 \geq 0, \beta_0 \geq \theta_G.\\
    \end{eqnarray*}
    To solve the optimization, we specify the Lagrangian as follows for some multipliers $\lambda \in \mathbb{R}^4$ and a slack variables $s\in \mathbb{R}^4$.
    \begin{eqnarray*}
        \mathcal{L}  := &  \delta r^T \gamma_1 - \gamma_0^T C_0 \gamma_0 - \lambda_1 (\alpha_0-s_1^2)  - \lambda_2(\beta_1 - \theta_G - s_2^2) - \lambda_3(U_D-s_3^2) - \lambda_4(U_G-s_4^2).
    \end{eqnarray*}
    We partially differentiate with respect to each decision variable and each multiplier.
    \begin{eqnarray*}
    \frac{\partial}{\partial \alpha_1} \mathcal{L} &=& 0\\
     \iff \delta r_\alpha - 2c_{0,\alpha\alpha}\alpha_0 + 2c_{0,\alpha\beta} \beta_0- \lambda_1 - \lambda_3 \frac{\partial U_G}{\partial \alpha_0}- \lambda_4 \frac{\partial U_D}{\partial \alpha_0} &=& 0 \\
     \iff (1-\lambda_3)\left(\delta r_\alpha - 2c_{0,\alpha\alpha}\alpha_0 + 2c_{0,\alpha\beta} \beta_0 \right) - \lambda_1 &&\\- \lambda_4\left((1-\delta) r_\alpha - 2c_{1,\alpha\alpha} (\alpha_1-\alpha_0) + 2c_{1,\alpha\beta} (\beta_1 - \max(\beta_0,\theta_D))\right)&=& 0, \\
    \frac{\partial}{\partial \beta_1} \mathcal{L} &=& 0\\
     \iff \delta r_\beta - 2c_{0,\beta\beta} \beta_0 + 2c_{0,\alpha\beta} \alpha_0- \lambda_2 - \lambda_3 \frac{\partial U_G}{\partial \beta_0}  - \lambda_4 \frac{\partial U_D}{\partial \beta_0} &=& 0 \\
     \iff (1-\lambda_3) \left(\delta r_\beta - 2c_{0,\beta\beta} \beta_0 + 2c_{0,\alpha\beta} \alpha_0\right)- \lambda_2 &&\\
     - \lambda_4\left((1-\delta) r_\beta - 2c_{1,\beta\beta} (\beta_1 - \max(\beta_0,\theta_D) + 2c_{1,\alpha\beta} (\alpha_1-\alpha_0)\right)&=& 0, \\
    \frac{\partial}{\partial \lambda_1} \mathcal{L} &=& 0 \\
    \iff -\alpha_0 + s_1^2  &=& 0, \\
    \frac{\partial}{\partial \lambda_2} \mathcal{L} &=& 0 \\
    \iff -\beta_0 + s_2^2  &=& 0, \\
    \frac{\partial}{\partial \lambda_3} \mathcal{L} &=& 0 \\
    \iff -U_G + s_2^2  &=& 0 \\
    \iff -\delta r^T \gamma_1 + \gamma_0^T C_1 \gamma_0 + s_2^2  &=& 0, \\
    \frac{\partial}{\partial \lambda_3} \mathcal{L} &=& 0 \\
    \iff -U_D + s_2^2  &=& 0 \\
    \iff -(1-\delta) r^T \gamma_1 + \left(\gamma_1-\gamma_0\right)^T C_1 \left(\gamma_1-\gamma_0\right) + s_2^2  &=& 0. \\
\end{eqnarray*}

Using complementary slackness, we have sixteen possible options. For brevity, we refer to these options by the constraints they satisfy, where \textbf{bold} corresponds to the constraints being activated. The algebra is omitted for exposition; only the candidates yielded are noted for each constraint setting.
\begin{enumerate}
    \item $\boldsymbol{\alpha_0,\beta_0,U_G, U_D} \rightarrow [0,\theta_G]$.
    \item $\boldsymbol{\alpha_0,\beta_0,U_G}, U_D \rightarrow [0,\theta_G]$  
    \item $\boldsymbol{\alpha_0,\beta_0},U_G, U_D \rightarrow [0,\theta_G]$
    \item $\boldsymbol{\alpha_0,\beta_0},U_G, \boldsymbol{U_D} \rightarrow [0,\theta_G]$ 
    \item $\boldsymbol{\alpha_0},\beta_0,\boldsymbol{U_G, U_D} \rightarrow \left[\begin{array}{c}
         0 \\
         \frac{\delta r_\beta}{2c_{0\beta\beta}}
    \end{array}\right]$ % \gamma_0^*=[0,0]$.
    \item $\boldsymbol{\alpha_0},\beta_0,\boldsymbol{U_G}, U_D \rightarrow \left[\begin{array}{c}
         0 \\
         \frac{\delta r_\beta}{2c_{0\beta\beta}}
    \end{array}\right]$ % \gamma_0^*=[0,0]$.
    \item $\boldsymbol{\alpha_0},\beta_0,U_G, U_D \rightarrow \left[\begin{array}{c}
         0 \\
         \frac{\delta r_\beta}{2c_{0\beta\beta}}
    \end{array}\right]$ % \gamma_0^*=[0,0]$.
    \item $\boldsymbol{\alpha_0},\beta_0,U_G, \boldsymbol{U_D} \rightarrow $ One of three along $U_D=0$ curve.% \gamma_0^*=[0,0]$.
    \item $\alpha_0,\boldsymbol{\beta_0,U_G, U_D} \rightarrow \gamma_0^*=\left[\begin{array}{c}
         \frac{\delta r_\alpha}{2c_{0\alpha\alpha}}-\frac{c_{0\alpha\beta}}{c_{0\alpha\alpha}}\theta_G\\
         \theta_G
    \end{array}\right]$.
    \item $\alpha_0,\boldsymbol{\beta_0,U_G}, U_D \rightarrow \left[\begin{array}{c}
         \frac{\delta r_\alpha}{2c_{0\alpha\alpha}}-\frac{c_{0\alpha\beta}}{c_{0\alpha\alpha}}\theta_G\\
         \theta_G
    \end{array}\right]$  % \gamma_0^*=[0,0]$.
    \item $\alpha_0,\boldsymbol{\beta_0},U_G, U_D \rightarrow \left[\begin{array}{c}
         \frac{\delta r_\alpha}{2c_{0\alpha\alpha}}-\frac{c_{0\alpha\beta}}{c_{0\alpha\alpha}}\theta_G\\
         \theta_G
    \end{array}\right]$ % \gamma_0^*=[0,0]$.
    \item $\alpha_0,\boldsymbol{\beta_0},U_G, \boldsymbol{U_D} \rightarrow $ Two of three along the $U_D=0$ curve. % \gamma_0^*=[0,0]$.
    \item $\alpha_0,\beta_0,\boldsymbol{U_G, U_D} \rightarrow \frac{\delta}{2}C_0^{-1} r$ .
    \item $\alpha_0,\beta_0,\boldsymbol{U_G}, U_D \rightarrow \frac{\delta}{2}C_0^{-1} r$  
    \item $\alpha_0,\beta_0,U_G, U_D \rightarrow \frac{\delta}{2}C_0^{-1} r$
    \item $\alpha_0,\beta_0,U_G, \boldsymbol{U_D} \rightarrow $ Three of three along the $U_D=0$ curve. 
\end{enumerate}
Thus we have established our four candidates in the proposition statement.
To handle the \texttt{abstain} scenario, we check each candidate produced in the process above by plugging the strategy to our formula for $U_G$. If none yield positive utility, then the generalist prefers to \texttt{abstain}.
\end{proof}

\section{Helper Lemmas and Analysis}

Here we write out helper Lemmas and analysis for our proofs concerning backfiring and mutualism.

\begin{lemma}
    In the AI regulation game with quadratic costs, any player's utility is strictly concave if and only if their cost matrix is positive definite. 
    \label{lemma:pd-iff-concave}
\end{lemma}
\begin{proof}
    The generalist utility function is given by $U_G = \delta r^T \gamma_1 - \gamma_0^T C_0 \gamma_0$. Observe this is twice differentiable everywhere. Thus the function is strictly concave in $\alpha_0,\beta_0$ if and only if its Hessian derivative is negative definite. We compute the Hessian as follows:
$$
H := \left[\begin{array}{cc}
\frac{\partial^2 U_G}{\partial \alpha_0^2} & \frac{\partial^2 U_G}{\partial \alpha_0 \partial \beta_0} \\
\frac{\partial^2 U_G}{\partial \beta_0 \partial \alpha_0} & \frac{\partial^2 U_G}{\partial \beta_0^2}
\end{array}\right] = -2C_0.
$$
This matrix is negative definite if and only if $C_0$ is positive definite.

The proof for the domain specialist follows the same steps.
\end{proof}

\begin{lemma}
    In any AI regulation game with separable quadratic costs, if there is no regulation, both players will invest a non-zero amount in each attribute.
\end{lemma}

\begin{proof}
    By Lemma \ref{lemma:pd-iff-concave}, we are given that the utilities are strictly concave. Thus, the proof consists of showing that 1) the utility function is greater than or equal to 0 at the origin point of zero investment and 2) the gradient points towards the interior of the feasible set everywhere along the boundaries.

    Here we prove the two conditions for $U_G$:
    \begin{enumerate}
        \item $U_G(\alpha_0=0,\beta_0=0) = \delta r^T \vec{0} - 0 = 0$
        \item We prove this for each constraint, $\alpha_0\geq 0, \beta_0\geq 0$:
        \begin{itemize}
            \item $\frac{\partial U_G}{\partial \alpha_0}\big|_{\beta_0=0} = \delta r_\alpha - 2c_{0,\alpha\alpha}\alpha_0 = \delta r_\alpha-0 > 0$.
            \item $\frac{\partial U_G}{\partial \beta_0}\big|_{\alpha_0=0} = \delta r_\beta - 2c_{0,\beta\beta}\beta_0 = \delta r_\beta-0 > 0$.
        \end{itemize}
    \end{enumerate}

    Here we prove the two conditions for $U_D$:
    \begin{enumerate}
        \item $U_D(\alpha_i=0,\beta_i=0) = (1-\delta) r^T \gamma_0 - 0 \geq 0$
        \item We prove this for each constraint, $\alpha_i\geq \alpha_0, \beta_i\geq \beta_0$:
        \begin{itemize}
            \item $\frac{\partial U_D}{\partial \alpha_i}\big|_{\beta_i=\beta_0} = (1-\delta) r_\alpha - 2c_{i,\alpha\alpha}(\alpha_1-\alpha_0) = (1-\delta) r_\alpha > 0$.
            \item $\frac{\partial U_G}{\partial \beta_0}\big|_{\alpha_i=\alpha_0} = (1-\delta) r_\beta - 2c_{i,\beta\beta}(\beta_i-\beta_0) = \delta r_\beta-0 > 0$.
        \end{itemize}
    \end{enumerate}
\end{proof}

\section{Proving the Backfiring Result}
\label{app:backfiring-proof}

Below we prove the Theorem \ref{thm:backfiring}. 

\begin{proof}
    Assume $\theta_G = 0$ for the entire proof. By Remark \ref{remark:unconstrained-condition}, we're given that the players commit to their unconstrained strategy in equilibrium. These were solved in Propositions \ref{prop:D-strategy-regulation} and \ref{prop:G-strategy-regulation}. Thus we have the following player's strategies under no regulation for this setting:
    \begin{equation}
        \gamma_0^A = \frac{\delta}{2}C_0^{-1}r,\ \ \  \gamma_1^A = \frac{1-\delta}{2}C_1^{-1}r.
    \label{eqn:backfiring-noreg-strategies}
    \end{equation}
    Our strategy is to show that G's unconstrained, no-regulation optimum becomes dominated in the presence of regulation targeting D, which we choose to be arbitrarily close to $\beta_1^A$.

    \textbf{Notation.} Before we proceed, we introduce some additional notation. Define the set $S$ to be all feasible pairs of strategies $(\gamma_0,\gamma_1)$. `Feasible' here means those strategies which leave both $G$ and $D$ with non-negative utility. We use the subscript $S_{\theta_D}$ to track the particular regulatory threshold. The feasible pairs of strategies in the unregulated game is given by $S_0$, and the feasible pairs of strategies in a game with threshold $\theta_D=1.5$ is denoted $S_{1.5}$. We may refer to the unregulated game with the superscript $A$ (for anarchy), e.g. $\beta_1^A$ refers to the unregulated safety level. Observe that any set of tuples $S_\theta$ can be separated into two mutually exclusive and collectively exhaustive sets: 
\begin{itemize}
    \item $S_\theta^{\text{MC}}$ (for minimally compliant) is the set of all tuples where $D$'s best response has safety $\beta_1^*=\theta_D$.
    \item $S_\theta^{\text{C}}$ (for contribute) is the set of all tuples where $D$'s best response has safety $\beta_1^*>\theta_D$.
\end{itemize}
Now, we provide a sequence of lemmas, with the purpose of establishing the intuition that \textit{all we must do is find some $\epsilon>0$ and some strategy $\beta_0^R\neq \beta_0^A$ such that G prefers $\beta_0^R$ to $\beta_0^A$ and D minimally complies}. 
\begin{lemma}
    For any threshold $\theta_D>0$, $S_\theta^{\text{C}} \subset S_0$.
    \label{obs:strat-subset}
\end{lemma}
\begin{proof}
$S_0 = S_0^{\text{MC}} \cup S_0^{\text{C}} = S_0^{\text{MC}} \cup \left(\bigcup_{t=0}^\infty S_t^{\text{C}} \right) \supset S_\theta^{\text{C}}$.   
\end{proof}

\begin{lemma}
    If $\theta_D \geq \beta_1^A$, backfiring is impossible.
    \label{lemma:backfiring-impossibility}
\end{lemma}
\begin{proof}
    Assume for contradiction that $\theta_D^* \geq \beta_1^A$ and backfiring occurs. Backfiring would imply $\beta_1(\theta_D=\theta_D^*) < \beta_1(\theta_D=0)= \beta_1^A$. However, this would violate the regulation, which we're given is greater than $\beta_1^A$. Hence we've already established the contradiction.
\end{proof}
\begin{lemma}
    Given a threshold $\theta$, backfiring can occur only if the strategies $(\gamma_0,\gamma_1) \in S_{\theta_D}^{\text{MC}}$. 
\end{lemma}
\begin{proof}
We have established $S_{\theta_D} = S_{\theta_D}^{\text{MC}} \cup S_{\theta_D}^{\text{C}}$, so the proof will show that the strategies in $S_\theta^{\text{C}}$ can never exhibit backfiring. This would imply, if backfiring occurs over the feasible set of strategies $S_\theta$, it is only possible for strategies in $S_\theta^{\text{MC}}$. The proof proceeds, first for all values $\theta_D \geq \beta_1^A$, and then for all values $\theta < \beta_1^A$.
\begin{itemize}
    \item For $\theta_D \geq \beta_1^A$, backfiring is impossible generally, as established in  Lemma \ref{lemma:backfiring-impossibility}.
    \item For $\theta_D < \beta_1^A$, start by observing that the anarchy solution  $(\gamma_0^A,\gamma_1^A)$ is always feasible. This solution is the strategy tuple that maximizes the utility of $G$ over $S_0$.  Lemma \ref{obs:strat-subset} tells us that this set, $S_0$, contains all sets of regulated strategies where the players contribute: $S_{\theta_D}^\text{contribute} \subset S_0$. Thus: $(\gamma_0^A,\gamma_1^A) := \sup_{U_G}S_0 \succeq_G S_0\supset  S_{\theta_D}^\text{contribute} \rightarrow (\gamma_0^A,\gamma_1^A) \succeq_G S_{\theta_D}^\text{C}$. Thus the anarchy solution is feasible and dominates all strategies in $S_{\theta_D}^{\text{contribute}}$.
\end{itemize}
This completes the proof, and demonstrates that if backfiring is ever to occur, it will exhibit strategies that are \textit{minimally compliant} with the regulation.
\end{proof}

Backfiring is a regulation yielding lower safety than $\beta_0^A$. The claims above state that backfiring cannot occur if $\theta_D>\beta_1^A$ and can only occur if the domain specialist minimally complies. As an immediate corollary, we can claim that backfiring occurs \textit{if and only if} there is a regulation $\theta_D<\beta_1^A$ such that the strategies $(\gamma_0(\theta_D), \gamma_1(\theta_D)) \in S_{\theta_D}^{MC}$. 

\begin{lemma}
    For a given regulation $\theta_D<\beta_1^A$ in our setting, if $G$ prefers \textit{any} minimally compliant strategy $\gamma_0'$ to $\gamma_0^A$, then $G$'s optimal strategy $\gamma_0^* \in S_{\theta_D}^{MC}$ and the regulation backfires.
\label{lemma:backfiring-suffcond}
\end{lemma}
\begin{proof}
    We're given $\gamma_0^A$ is optimal over $S_0$. By Lemma \ref{obs:strat-subset}, $S_{\theta_D}^C \subset S_0$. Since $\theta_D<\beta_1^A$, $\gamma_0^A$ remains feasible. The only new strategies available to $G$ are those in $S_{\theta_D}^{MC}$. Thus, if we denote utility-domination using $\succ$, we have $\gamma_0' \succ \gamma_0^A \succeq g \forall g\in S_{\theta_D}^C$. This implies G's optimal strategy $\gamma_0^*$ is either $\gamma_0'$ or otherwise belongs to $S_{\theta_D}^{MC}$.
\end{proof}

Thus our task is to find some regulation $\theta_D$ and some strategy $\gamma_0'$ such that $U_G(\gamma_0')>U_G(\gamma_0^A)$. 

\begin{lemma}
    For small $\epsilon>0$, if the given conditions are met, the following $G$ strategy dominates no regulation:
    $$\gamma_0' = \left[\begin{array}{c}
        \frac{\delta r_\alpha}{2c_{0,\alpha\alpha}}\left(\beta_0^A - 2\epsilon\right) \\
        \beta_0^A-2\epsilon
    \end{array}\right]$$
\end{lemma}
\begin{proof}
    Equation \ref{eqn:backfiring-noreg-strategies} give us $G$ and $D$'s strategies under no regulation. Given $G$'s candidate strategy stated in the Lemma, we compute $D$'s best response. Observe this must be a minimally-compliant best response, because $G$'s strategy was constructed to be a difference $\beta_0^A + \epsilon$ from D's regulatory floor. Thus, by Proposition \ref{prop:D-strategy-regulation}, we have:
$$\gamma_1' = \left[\begin{array}{c}
    \alpha_0' + \frac{(1-\delta)}{2c_{1\alpha\alpha}}-\frac{c_{1,\alpha\beta}}{c_{1\alpha\alpha}}\theta_D  \\
    \theta_D
\end{array}\right]$$
We compare G's utility in the two scenarios:
\begin{enumerate}
    \item $(\gamma_0',\gamma_1') \rightarrow     U_G' = \delta \left( r_\alpha \alpha_1' + r_\beta \beta_1' \right) - c_{0,\alpha\alpha} (\alpha_0')^2 - 2c_{0,\alpha\beta} \alpha_0' \beta_0' - c_{0,\beta\beta} (\beta_0')^2$
    \item $(\gamma_1^A, \gamma_1^A) \rightarrow U_G^A = \delta \left( r_\alpha \alpha_1^A + r_\beta \beta_1^A \right) - c_{0,\alpha\alpha} (\alpha_0^A)^2 - 2c_{0,\alpha\beta} \alpha_0^A \beta_0^A - c_{0,\beta\beta} (\beta_0^A)^2$ 
\end{enumerate}
We compute the difference $\Delta U_G = U_G'-U_G^A$. We expand both terms and take the limit as $\epsilon \searrow 0$ to get the following:
$$\lim_{\epsilon \searrow 0}\Delta U_G = \frac{\delta (1-\delta)r_\beta}{2c_{1\beta\beta}}\left(r_{\beta} - \frac{c_{1,\alpha\beta}}{c_{1,\alpha\alpha}}r_\alpha \right)$$

A sufficient condition for this quantity being positive is stated below. The reason is all terms outside the parentheses are given as positive.
$$r_\beta > \frac{c_{1,\alpha\beta}}{c_{1,\alpha\alpha}}r_\alpha.$$
Notice the above condition is given as it is one of the conditions in remark \ref{remark:unconstrained-condition}.\footnote{This is also the condition for having a non-zero safety investment when costs are convex, and intuitively, backfiring is impossible when safety investment is zero.}
\end{proof}
     Thus, we have shown that for small positive $\epsilon$, the generalist prefers the backfiring strategy to the unconstrained optimum $\gamma_0^A$. By Lemma \ref{lemma:backfiring-suffcond}, the optimal regulated strategy is an element in $S_{\theta_D}^{MC}$ and the regulation backfires.
\end{proof}

\section{Proof of the mutualism result}
\label{app:mutualism}

Here we prove Theorem \ref{thm:mutualism}, that for a swath of games there exists a set of regluations that mutually improve the player's utilities.
\begin{proof}
    Observe that we only have to provide a single instance of regulation that does better than the unregulated optimal $\gamma_0^A,\gamma_1^A$ to show that there exists a Pareto improvement effect of regulation. We consider the following minimal-compliance strategies (using Proposition \ref{prop:D-strategy-regulation} and \ref{prop:G-strategy-regulation}):
$$\gamma_0' = \left[\begin{array}{c}
    \frac{\delta r_\alpha}{2c_{0,\alpha\alpha}}-\frac{c_{0,\alpha\beta}}{c_{0,\alpha\alpha}}\theta_G\\
     \theta_G
\end{array}\right], \gamma_1' = \left[\begin{array}{c}
    \alpha_0 + \frac{(1-\delta)r_\alpha}{2c_{1,\alpha\alpha}}-\frac{c_{1,\alpha\beta}}{c_1{\alpha\alpha}}\theta_D\\
     \theta_D
\end{array}\right].$$
Observe these are feasible because they are compliant and, for small $\epsilon$, the performance investment is positive.
    \begin{lemma}
        For the specified conditions, $U_G(\gamma_0',\gamma_1')> U_G(\gamma_0^A,\gamma_1^A)$
    \end{lemma}
    \begin{proof}
     Start by computing the change in the generalist's performance and safety investments between these strategies. The change in safety investment is simply $\Delta \beta_0 = \theta_G - \beta_0^A = \epsilon$. The change in performance investment is given by:
    \begin{eqnarray*}
        \Delta \alpha_0 =& \frac{\delta r_\alpha}{2c_{0,\alpha\alpha}}- \frac{c_{0,\alpha\beta}}{c_{0,\alpha\alpha}}\left(\frac{\delta}{2\det C_0} (-c_{0,\alpha\beta}r_\alpha + c_{0,\alpha\alpha}r_\beta + \epsilon)\right) - \frac{\delta}{2\det C_0}\left(c_{0\beta\beta} r_\alpha - c_{0\alpha\beta} r_\beta\right)\\ 
        = & \frac{\delta r_\alpha}{2c_{0,\alpha\alpha}} + \frac{c_{0,\alpha\beta}^2}{c_{0,\alpha\alpha}}\frac{\delta r_\alpha}{2\det C_0} - \cancel{c_{0,\alpha\beta}\frac{\delta}{2\det C_0}r_\beta}+ \frac{c_{0,\alpha\beta}}{c_{0,\alpha\alpha}}\epsilon - \frac{\delta}{2\det C_0}c_{0,\beta\beta}r_{\alpha} + \cancel{c_{0,\alpha\beta}\frac{\delta}{2\det C_0}r_\beta}\\
        = & \left(\frac{\delta}{2c_{0,\alpha\alpha}}+ \frac{c_{0,\alpha\beta}^2\delta}{c_{0,\alpha\alpha} 2 \det C_0}- \frac{\delta c_{0,\beta\beta}}{2 \det C_0 }\right)r_\alpha + \frac{c_{0,\alpha\beta}}{c_{0,\alpha\alpha}}\epsilon \\
        = & \frac{\delta r_\alpha}{2}\left(\frac{1}{c_{0,\alpha\alpha}}+\frac{c_{0,\alpha\beta}^2}{c_{0,\alpha\alpha}(c_{0,\alpha\alpha}c_{0,\beta\beta}-c_{0,\alpha\beta}^2)}-\frac{c_{0,\beta\beta}}{c_{0,\alpha\alpha}}\right) + \frac{c_{0,\alpha\beta}}{c_{0,\alpha\alpha}}\epsilon\\
        = &  \frac{\delta r_\alpha}{2}\cancel{\left(\frac{\det C_0 + c_{\alpha\beta}^2 - c_{0,\alpha\alpha}c_{0,\beta\beta}}{c_{0,\alpha\alpha}\det C_0}\right)} + \frac{c_{0,\alpha\beta}}{c_{0,\alpha\alpha}}\epsilon\\
        = &\frac{c_{0,\alpha\beta}}{c_{0,\alpha\alpha}}\epsilon.
    \end{eqnarray*}
    By the same logic, we solve for the change in the players' strategies. First, $\Delta \beta_1 = \beta_1' - \beta_1^A = \beta_1^A + 2\epsilon - \beta_1^A = 2\epsilon$. The change in performance is given by: 

        \begin{eqnarray*}
            \Delta \alpha_1 = & \alpha_1' - \alpha_1^A \\
            = & \left[\alpha_0' + \frac{(1-\delta)r_\alpha}{2c_{1,\alpha\alpha}}-\frac{c_{1,\alpha\beta}}{c_{1,\alpha\alpha}}\left(\theta_D - \beta_0' \right)\right]- \left[\alpha_0^A + \frac{(1-\delta)}{2\det C_1}(c_{1,\beta\beta}r_\alpha - c_{1,\alpha\beta}r_\beta)\right]\\
            = & \Delta \alpha_0 + \frac{(1-\delta)r_\alpha}{2c_{1,\alpha\alpha}}-\frac{c_{1,\alpha\beta}}{c_{1,\alpha\alpha}}\left(\frac{(1-\delta)}{2\det C_1}\left(-c_{1,\alpha\beta}r_\alpha + c_{1,\alpha\alpha}r_\beta\right) + \epsilon\right) - \frac{(1-\delta)}{2\det C_1}\left(c_{1,\beta\beta}r_\alpha - c_{1,\alpha\beta}r_\beta\right)\\
            = & \frac{c_{0,\alpha\beta}}{c_{0,\alpha\alpha}}\epsilon + \frac{(1-\delta)r_\alpha}{2c_{1,\alpha\alpha}}-\frac{c_{1,\alpha\beta}}{c_{1,\alpha\alpha}}\left(\frac{(1-\delta)}{2\det C_1}\left(-c_{1,\alpha\beta}r_\alpha + c_{1,\alpha\alpha}r_\beta\right) + \epsilon\right) - \frac{(1-\delta)}{2\det C_1}\left(c_{1,\beta\beta}r_\alpha - c_{1,\alpha\beta}r_\beta\right)\\
            = & \epsilon\left(\frac{c_{0,\alpha\beta}}{c_{0,\alpha\alpha}} - \frac{c_{1,\alpha\beta}}{c_{1,\alpha\alpha}}\right).
        \end{eqnarray*}
    The change in $G$'s cost is given by:
    \begin{eqnarray*}
        \Delta (\text{G's cost}) = & \left[\begin{array}{c}
            \Delta \alpha_0 \\
            \Delta \beta_0
        \end{array}\right]^T C_0 \left[\begin{array}{c}
            \Delta \alpha_0 \\
            \Delta \beta_0
        \end{array}\right] \\  
        = & c_{0,\alpha\alpha}(\frac{c_{0,\alpha\beta}}{c_{0,\alpha\alpha}})^2 \epsilon^2 + 2\left(\frac{c_{0,\alpha\beta}}{c_{0,\alpha\alpha}}\epsilon\right)\epsilon+c_{0\beta\beta}\epsilon^2
    \end{eqnarray*}
    Notice these are all $\epsilon^2$ terms, meaning as $\epsilon$ is brought to very small positive values, they approach zero at an exponential rate. The contribution to $G$'s revenue is given by:
    \begin{eqnarray*}
        \Delta (\text{G's revenue}) = & \delta(r_\alpha\Delta \alpha_1 + r_\beta\Delta \beta_1)\\
        = & \delta \left(r_\alpha \epsilon\left(\frac{c_{0,\alpha\beta}}{c_{0,\alpha\alpha}} - \frac{c_{1,\alpha\beta}}{c_{1,\alpha\alpha}}\right) + r_\beta 2\epsilon\right)
    \end{eqnarray*}
    Notice these are terms of $\epsilon$, whereas the cost effects are solely terms of $\epsilon^2$. Therefore, for sufficiently small $\epsilon$, we say:
    $$\lim_{\epsilon \searrow 0} \Delta U_G =\delta\left(r_\alpha \epsilon \left(\frac{c_{0,\alpha\beta}}{c_{0,\alpha\alpha}}-\frac{c_{1,\alpha\beta}}{c_{1,\alpha\alpha}}+2r_{\beta } \epsilon\right) \right)$$
    Using the given conditions, we know:
    \begin{eqnarray*}
        \lim_{\epsilon \searrow 0} \Delta U_G  >0 \iff & \delta\left(r_\alpha \epsilon \left(\frac{c_{0,\alpha\beta}}{c_{0,\alpha\alpha}}-\frac{c_{1,\alpha\beta}}{c_{1,\alpha\alpha}}\right)+2r_{\beta } \epsilon\right) >0\\
        \iff & r_\alpha  \left(\frac{c_{0,\alpha\beta}}{c_{0,\alpha\alpha}}-\frac{c_{1,\alpha\beta}}{c_{1,\alpha\alpha}}\right) +2r_{\beta } >0 \\
        %\iff & r_\alpha  \left(\frac{c_{0,\alpha\beta}}{c_{0,\alpha\alpha}}-\frac{c_{1,\alpha\beta}}{c_{1,\alpha\alpha}}\right) +2r_{\beta }>0\\
        \iff & \frac{r_\alpha c_{0,\alpha\beta}}{r_\beta c_{0,\alpha\alpha}} - \frac{r_\alpha c_{1,\alpha\beta}}{r_\beta c_{1,\alpha\alpha}} > -2.
    \end{eqnarray*}
    Our conditions strictly bound the absolute value of both terms on the left hand side below 1, so this completes the Lemma's proof.
    \end{proof}
    \begin{lemma}
        For the specified conditions, $U_D(\gamma_0',\gamma_1')> U_D(\gamma_0^A,\gamma_1^A)$
    \end{lemma}
        The limiting effect on $D$'s revenue is calculated exactly the same way as above, except that the revenue expression is multiplied by $(1-\delta)$ instead of $\delta$.

    This completes the proof, as both players are better off under the regulation.
\end{proof}

\end{document}